\documentclass[11pt,twoside,a4paper,notitlepage]{article}

\usepackage{ifthen}

\newboolean{publisherversion}
\setboolean{publisherversion}{false}

\ifthenelse{\boolean{publisherversion}}
{}
{
  \usepackage[a4paper,margin=2.5cm]{geometry}
  \usepackage{tgtermes}
  \usepackage[scale=.92]{tgheros}
  \usepackage{tgcursor}
}

\ifthenelse{\boolean{publisherversion}}
{}
{
}

\ifthenelse{\boolean{publisherversion}}
{}
{
  \usepackage[english]{babel}
}

\ifthenelse{\boolean{publisherversion}}
{}
{
  \usepackage{fancyhdr}
}

\ifthenelse{\boolean{publisherversion}}
{}
{
  
}

\usepackage{xcolor}
\usepackage{braket}
\usepackage{bussproofs}
\usepackage{algorithm}
\usepackage[noend]{algpseudocode}
\usepackage{fancyvrb}
\usepackage{textcomp}
\newcommand{\mytildesymbol}{\raisebox{0.60ex}{\texttildelow}}
\newcommand{\verbtilde}[1]{\mytildesymbol#1}

\usepackage{tabularx}

\usepackage{ifthen}

\usepackage{ifthen}

\provideboolean{detectedSTOC}
\provideboolean{detectedFOCS}
\provideboolean{detectedElsevier}
\provideboolean{detectedNOW}
\provideboolean{detectedLMCS}
\provideboolean{detectedIEEE}
\provideboolean{detectedPoster}
\provideboolean{detectedSIAM}
\provideboolean{detectedLNCS}
\provideboolean{detectedACM}
\provideboolean{detectedACMconf}
\provideboolean{detectedSigplanconf}
\provideboolean{detectedToC}
\provideboolean{detectedLIPIcs}
\provideboolean{detectedAAAI}
\provideboolean{detectedIJCAI}
\provideboolean{detectedCompCplx}
\provideboolean{detectedEasyChair}
\provideboolean{detectedJAIR}
\provideboolean{detectedArticle}
\provideboolean{detectedReport}
\provideboolean{detectedThesis}

\makeatletter

\@ifclassloaded{sig-alternate}
{\setboolean{detectedSTOC}{true}}
{\setboolean{detectedSTOC}{false}}

\@ifclassloaded{elsarticle}
{\setboolean{detectedElsevier}{true}}
{\setboolean{detectedElsevier}{false}}

\@ifclassloaded{now}
{\setboolean{detectedNOW}{true}}
{\setboolean{detectedNOW}{false}}

\@ifclassloaded{lmcs}
{\setboolean{detectedLMCS}{true}}
{\setboolean{detectedLMCS}{false}}

\@ifclassloaded{IEEEtran} {
  \setboolean{detectedIEEE}{true}
  \ifCLASSOPTIONconference {
    \setboolean{detectedFOCS}{true}
  }
  \else {
    \setboolean{detectedFOCS}{false}
  }
  \fi
}
{
  \setboolean{detectedFOCS}{false}
  \setboolean{detectedIEEE}{false}
}

\@ifclassloaded{siamart171218}
{\setboolean{detectedSIAM}{true}}
{\setboolean{detectedSIAM}{false}}

\@ifclassloaded{llncs}
{\setboolean{detectedLNCS}{true}}
{\setboolean{detectedLNCS}{false}}

\@ifclassloaded{acmsmall}
{\setboolean{detectedACM}{true}}
{\setboolean{detectedACM}{false}}

\@ifclassloaded{acmart}
{\setboolean{detectedACMconf}{true}}
{\setboolean{detectedACMconf}{false}}

\@ifclassloaded{sigplanconf}
{\setboolean{detectedSigplanconf}{true}}
{\setboolean{detectedSigplanconf}{false}}

\@ifclassloaded{toc}
{\setboolean{detectedToC}{true}}
{\setboolean{detectedToC}{false}}

\@ifclassloaded{lipics}
{\setboolean{detectedLIPIcs}{true}}
{\@ifclassloaded{lipics-v2019}
  {\setboolean{detectedLIPIcs}{true}}
  {\@ifclassloaded{oasics-v2019}
    {\setboolean{detectedLIPIcs}{true}}
    {\setboolean{detectedLIPIcs}{false}}
  }
}

\@ifclassloaded{cc}
{\setboolean{detectedCompCplx}{true}}
{\setboolean{detectedCompCplx}{false}}

\@ifclassloaded{easychair}
{\setboolean{detectedEasyChair}{true}}
{\setboolean{detectedEasyChair}{false}}

\@ifpackageloaded{jair}
{\setboolean{detectedJAIR}{true}}        
{\setboolean{detectedJAIR}{false}}        

\@ifpackageloaded{aaai}
{\setboolean{detectedAAAI}{true}}        
{\@ifpackageloaded{aaai18}
  {\setboolean{detectedAAAI}{true}}
  {\@ifpackageloaded{aaai20}       
    {\setboolean{detectedAAAI}{true}}
    {\setboolean{detectedAAAI}{false}}}}

\@ifpackageloaded{ijcai18}
{\setboolean{detectedIJCAI}{true}}        
{\@ifpackageloaded{ijcai19}
  {\setboolean{detectedIJCAI}{true}}        
  {\setboolean{detectedIJCAI}{false}}}

\@ifclassloaded{sciposter}
{\setboolean{detectedPoster}{true}}
{\setboolean{detectedPoster}{false}}

\@ifclassloaded{article}
{\setboolean{detectedArticle}{true}}
{\setboolean{detectedArticle}{false}}

\makeatother

\ifthenelse{\not \isundefined{\examen} 
  \and \not \isundefined{\disputationsdatum} 
  \and \not \isundefined{\disputationslokal}}   
  {\setboolean{detectedThesis}{true}}
  {\setboolean{detectedThesis}{false}}

\ifthenelse{\boolean{detectedArticle} \or \boolean{detectedThesis}
  \or \boolean{detectedSTOC} \or \boolean{detectedFOCS}
  \or \boolean{detectedSIAM} \or \boolean{detectedIEEE}
  \or \boolean{detectedPoster}}
{\setboolean{detectedReport}{false}}
{\setboolean{detectedReport}{true}}

\ifthenelse{\boolean{detectedLIPIcs}}
{}
{\usepackage[utf8]{inputenc}}

\ifthenelse{\boolean{detectedIEEE}}
{\usepackage[cmex10]{amsmath}
  \interdisplaylinepenalty=2500}
{\usepackage{amsmath}}
\usepackage{amssymb}
\usepackage{amsfonts}

\usepackage{mathtools}

\usepackage{bussproofs}

\ifthenelse{\boolean{detectedJAIR}}
{}
{\usepackage{microtype}}

\usepackage{bm}  

\usepackage[style=base]{caption}  %
\usepackage{subcaption}
\ifthenelse
{     \boolean{detectedElsevier} \or \boolean{detectedSIAM}
  \or \boolean{detectedSIAM}     \or \boolean{detectedLIPIcs}}
{}
{\usepackage{varioref}}

\usepackage{xspace}

\ifthenelse    
{\boolean{detectedSTOC}      \or \boolean{detectedSIAM}         \or 
  \boolean{detectedLMCS}     \or \boolean{detectedNOW}          \or 
  \boolean{detectedLNCS}     \or \boolean{detectedACM}          \or
  \boolean{detectedToC}      \or \boolean{detectedLIPIcs}       \or
  \boolean{detectedACMconf}  \or \boolean{detectedAAAI}         \or
  \boolean{detectedCompCplx} \or \boolean{detectedSigplanconf}  \or
  \boolean{detectedJAIR}}
{}
{
  \usepackage{amsthm}
  \usepackage{amsthm-boldfont}
}

\ifthenelse        
{\boolean{detectedElsevier}  \or \boolean{detectedFOCS}         \or 
  \boolean{detectedSTOC}     \or \boolean{detectedPoster}       \or
  \boolean{detectedSIAM}     \or \boolean{detectedLMCS}         \or
  \boolean{detectedIEEE}     \or \boolean{detectedNOW}          \or
  \boolean{detectedToC}      \or \boolean{detectedThesis}       \or
  \boolean{detectedLNCS}     \or \boolean{detectedACM}          \or 
  \boolean{detectedLIPIcs}   \or \boolean{detectedAAAI}         \or
  \boolean{detectedACMconf}  \or \boolean{detectedIJCAI}        \or 
  \boolean{detectedCompCplx} \or \boolean{detectedSigplanconf}  \or
  \boolean{detectedJAIR}}
{}
{\usepackage{jncaptionsheadings}}

\ifthenelse{\isundefined{\DONOTLOADHYPERREFPACKAGE}
  \and \not \boolean{detectedSTOC}        \and \not \boolean{detectedFOCS}
  \and \not \boolean{detectedPoster}      \and \not \boolean{detectedElsevier} 
  \and \not \boolean{detectedSIAM}        \and \not \boolean{detectedACM}
  \and \not \boolean{detectedIEEE}        \and \not \boolean{detectedNOW}
  \and \not \boolean{detectedToC}         \and \not \boolean{detectedThesis}
  \and \not \boolean{detectedLIPIcs}      \and \not \boolean{detectedSIAM}
  \and \not \boolean{detectedAAAI}        \and \not \boolean{detectedIJCAI}
  \and \not \boolean{detectedSigplanconf} \and \not \boolean{detectedACMconf}   
  \and \not \boolean{detectedCompCplx} \and \not \boolean{detectedEasyChair}}
{\usepackage[colorlinks=true,citecolor=blue]{hyperref}}
{}

\ifthenelse{\boolean{detectedSTOC}}                  
{%

\newtheorem{theorem}{Theorem}
\newtheorem{lemma}[theorem]{Lemma}
\newtheorem{proposition}[theorem]{Proposition}
\newtheorem{corollary}[theorem]{Corollary}
\newtheorem{observation}[theorem]{Observation}

\newtheorem{definition}[theorem]{Definition}
\newtheorem{claim}[theorem]{Claim}

\newtheorem{conjecture}{Conjecture}
\newtheorem{openproblem}[conjecture]{Open Problem}

\newdef{remark}{Remark}
\newdef{example}{Example}

\newcounter{unnumber}

}
{\ifthenelse{\boolean{detectedSIAM}}
  {%

\newsiamremark{remark}{Remark}
\newsiamremark{hypothesis}{Hypothesis}

\newsiamthm{claim}{Claim}
\newsiamthm{observation}{Observation}
\newsiamthm{conjecture}{Conjecture}

\crefname{hypothesis}{Hypothesis}{Hypotheses}

\crefname{section}{Section}{Sections}
\crefname{claim}{Claim}{Claims}

\renewcommand{\refsec}[1]{\cref{#1}}

\renewcommand{\reffig}[1]{\cref{#1}}

}
  {\ifthenelse{\boolean{detectedNOW}}
    {%
\newtheorem{standardlocalcounter}{Dummy}[chapter]
\newtheorem{standardglobalcounter}{Dummy}

\newtheorem{theorem}[standardlocalcounter]{Theorem}
\newtheorem{lemma}[standardlocalcounter]{Lemma}
\newtheorem{proposition}[standardlocalcounter]{Proposition}
\newtheorem{corollary}[standardlocalcounter]{Corollary}
\newtheorem{observation}[standardlocalcounter]{Observation}
\newtheorem{fact}[standardlocalcounter]{Fact}

\newtheorem{conjecturelocalcounter}[standardlocalcounter]{Conjecture}
\newtheorem{conjectureglobalcounter}[standardglobalcounter]{Conjecture}
\newtheorem{conjecture}[standardglobalcounter]{Conjecture}
\newtheorem{openquestion}[standardglobalcounter]{Open Question}
\newtheorem{openproblem}[standardglobalcounter]{Open Problem}
\newtheorem{problem}{Problem}

\newtheorem{property}[standardlocalcounter]{Property}
\newtheorem{definition}[standardlocalcounter]{Definition}
\newtheorem{claim}[standardlocalcounter]{Claim}

\newtheorem{algorithm}[standardlocalcounter]{Algorithm}

\newtheorem{remark}[standardlocalcounter]{Remark}
\newtheorem{example}[standardlocalcounter]{Example}

\renewenvironment{proof}[1][Proof]{\par\trivlist
   \item[\hskip \labelsep{\itshape {#1}.}]\prooffont}
   {\hspace*{0pt plus1fill}\fboxsep2.5pt\fboxrule.5pt\raise3pt\hbox{\fbox{}}\endtrivlist}
}
    {\ifthenelse{\boolean{detectedLMCS}}
      {%

\theoremstyle{plain}    

\newtheorem{theorem}[thm]{Theorem}
\newtheorem{lemma}[thm]{Lemma}
\newtheorem{proposition}[thm]{Proposition}
\newtheorem{corollary}[thm]{Corollary}
\newtheorem{observation}[thm]{Observation}

\newtheorem{conjecture}[thm]{Conjecture}
\newtheorem{problem}[thm]{Problem}

\newtheorem{openquestion}{Open Question}
\newtheorem{openproblem}{Open Problem}

\theoremstyle{definition}

\newtheorem{property}[thm]{Property}
\newtheorem{definition}[thm]{Definition}
\newtheorem{claim}[thm]{Claim}
\newtheorem{remark}[thm]{Remark}
\newtheorem{example}[thm]{Example}

}
      {\ifthenelse{\boolean{detectedIEEE}}
        {%
\newtheorem{standardlocalcounter}{Dummy}[section]
\newtheorem{standardglobalcounter}{Dummy}

\theoremstyle{plain}    

\newtheorem{theorem}[standardglobalcounter]{Theorem}
\newtheorem{lemma}[standardglobalcounter]{Lemma}
\newtheorem{proposition}[standardglobalcounter]{Proposition}
\newtheorem{corollary}[standardglobalcounter]{Corollary}
\newtheorem{observation}[standardglobalcounter]{Observation}
\newtheorem{fact}[standardglobalcounter]{Fact}

\newtheorem{conjecture}[standardglobalcounter]{Conjecture}
\newtheorem{openquestion}{Open Question}
\newtheorem{openproblem}{Open Problem}
\newtheorem{problem}{Problem}

\theoremstyle{definition}

\newtheorem{property}[standardglobalcounter]{Property}
\newtheorem{definition}[standardglobalcounter]{Definition}
\newtheorem{claim}[standardglobalcounter]{Claim}

\theoremstyle{remark}
\newtheorem{remark}[standardglobalcounter]{Remark}
\newtheorem{example}[standardglobalcounter]{Example}

\newtheoremstyle{meta}%
  {3pt}%
  {3pt}%
  {\scshape \small }%
  {}%
  {\scshape \small }%
  {:}%
  { }%
  {}%

\theoremstyle{meta}
\newtheorem{meta}{Meta comment}

\newtheoremstyle{questions}%
  {3pt}%
  {3pt}%
  {\sffamily \slshape}%
  {}%
  {\bfseries \sffamily \slshape}%
  {:}%
  { }%
  {}%

\theoremstyle{questions}
\newtheorem{questions}{Open questions}

}
        {\ifthenelse{\boolean{detectedLNCS}}
          {%

\spnewtheorem*{proofsketch}{Proof sketch}{\itshape}{\rmfamily}

\spnewtheorem{observation}{Observation}{\bfseries}{\itshape}
\spnewtheorem{fact}{Fact}{\bfseries}{\itshape}

}
          {\ifthenelse{\boolean{detectedACM} \or \boolean{detectedACMconf}}
            {%

\theoremstyle{acmplain}
\newtheorem{theorem}{Theorem}[section]        

\newtheorem{observation}[theorem]{Observation}
\newtheorem{fact}[theorem]{Fact}
\newtheorem{claim}[theorem]{Claim}
\newtheorem{property}[theorem]{Property}
\newtheorem{subclaim}[theorem]{Subclaim}
\newtheorem{openquestion}{Open Question}
\newtheorem{openproblem}{Open Problem}

}
            {\ifthenelse{\boolean{detectedToC}}
              {%

\theoremstyle{plain}
\newtheorem{observation}[theorem]{Observation}
\newtheorem{openproblem}[theorem]{Open Problem}

\theoremstyle{definition}
\newtheorem{property}[theorem]{Property}

\renewcommand{\refsec}[1]{\expref{Section}{#1}}

\renewcommand{\reffig}[1]{\expref{Figure}{#1}}

}
              {\ifthenelse{\boolean{detectedLIPIcs}}
                {%

\theoremstyle{plain}    

\newtheorem{fact}[theorem]{Fact}
\newtheorem{observation}[theorem]{Observation}
}
                {\ifthenelse{\boolean{detectedCompCplx}}
                  {%
}
                  {\ifthenelse{\boolean{detectedSigplanconf}}
                    {%

\usepackage{amsthm}          
\usepackage{amsthm-boldfont}
\usepackage{url}

\newtheorem{standardlocalcounter}{Dummy}[section]
\newtheorem{standardglobalcounter}{Dummy}

\theoremstyle{plain}    

\newtheorem{theorem}[standardlocalcounter]{Theorem}
\newtheorem{lemma}[standardlocalcounter]{Lemma}
\newtheorem{proposition}[standardlocalcounter]{Proposition}
\newtheorem{corollary}[standardlocalcounter]{Corollary}
\newtheorem{observation}[standardlocalcounter]{Observation}
\newtheorem{fact}[standardlocalcounter]{Fact}

\newtheorem{conjecturelocalcounter}[standardlocalcounter]{Conjecture}
\newtheorem{conjectureglobalcounter}[standardglobalcounter]{Conjecture}
\newtheorem{conjecture}[standardglobalcounter]{Conjecture}
\newtheorem{openquestion}[standardglobalcounter]{Open Question}
\newtheorem{openproblem}[standardglobalcounter]{Open Problem}
\newtheorem{problem}[standardglobalcounter]{Problem}
\newtheorem{question}[standardglobalcounter]{Question}

\theoremstyle{definition}

\newtheorem{property}[standardlocalcounter]{Property}
\newtheorem{definition}[standardlocalcounter]{Definition}
\newtheorem{claim}[standardlocalcounter]{Claim}
\newtheorem{subclaim}[standardlocalcounter]{Subclaim}

\newtheorem{algorithm}[standardlocalcounter]{Algorithm}

\theoremstyle{remark}
\newtheorem{remark}[standardlocalcounter]{Remark}
\newtheorem{example}[standardlocalcounter]{Example}

}
                    {\ifthenelse{\boolean{detectedAAAI}}
                      {}
                      {\ifthenelse{\boolean{detectedArticle} 
                          \or \boolean{detectedElsevier}
                          \or \boolean{detectedEasyChair}}
                        {%
\newtheorem{standardlocalcounter}{Dummy}[section]
\newtheorem{standardglobalcounter}{Dummy}

\theoremstyle{plain}    

\newtheorem{theorem}[standardlocalcounter]{Theorem}
\newtheorem{lemma}[standardlocalcounter]{Lemma}
\newtheorem{proposition}[standardlocalcounter]{Proposition}
\newtheorem{corollary}[standardlocalcounter]{Corollary}
\newtheorem{observation}[standardlocalcounter]{Observation}

\newtheorem{conjecturelocalcounter}[standardlocalcounter]{Conjecture}
\newtheorem{conjectureglobalcounter}[standardglobalcounter]{Conjecture}
\newtheorem{conjecture}[standardglobalcounter]{Conjecture}
\newtheorem{openquestion}[standardglobalcounter]{Open Question}
\newtheorem{openproblem}[standardglobalcounter]{Open Problem}
\newtheorem{problem}[standardglobalcounter]{Problem}

\theoremstyle{definition}

\newtheorem{property}[standardlocalcounter]{Property}
\newtheorem{definition}[standardlocalcounter]{Definition}
\newtheorem{claim}[standardlocalcounter]{Claim}

\theoremstyle{remark}
\newtheorem{remark}[standardlocalcounter]{Remark}
\newtheorem{example}[standardlocalcounter]{Example}

}
                        {%
\newtheorem{standardlocalcounter}{Dummy}[chapter]
\newtheorem{standardglobalcounter}{Dummy}

\theoremstyle{plain}    

\newtheorem{theorem}[standardlocalcounter]{Theorem}

\newtheorem{proposition}[standardlocalcounter]{Proposition}

\theoremstyle{definition}

\theoremstyle{remark}

\newtheoremstyle{meta}%
  {3pt}%
  {3pt}%
  {\scshape \small }%
  {}%
  {\scshape \small }%
  {:}%
  { }%
  {}%

\theoremstyle{meta}

\newtheoremstyle{questions}%
  {3pt}%
  {3pt}%
  {\sffamily \slshape}%
  {}%
  {\bfseries \sffamily \slshape}%
  {:}%
  { }%
  {}%

\theoremstyle{questions}

}}}}}}}}}}}}}

\ifthenelse{\boolean{detectedThesis}}
{}
{\usepackage{ifpdf}}

\ifthenelse{\boolean{detectedToC}}
{}
{\usepackage{graphicx}}

\ifpdf         
\fi

\ifthenelse
{\boolean{detectedSTOC}   \or \boolean{detectedThesis} \or 
 \boolean{detectedLNCS}   \or \boolean{detectedToC}    \or 
 \boolean{detectedLIPIcs} \or \boolean{detectedAAAI}   \or
 \boolean{detectedIJCAI}  \or \boolean{detectedSIAM}}
{}
{
  \setcounter{secnumdepth}{3}
  \setcounter{tocdepth}{3}
}

\newcount\hours
\newcount\minutes
\def\SetTime{\hours=\time
\global\divide\hours by 60
\minutes=\hours
\multiply\minutes by 60
\advance\minutes by-\time
\global\multiply\minutes by-1 }
\SetTime
\def\now{\number\hours:\ifnum\minutes<10 0\fi\number\minutes}

\newcommand{\s}[1]{#1}

\newcommand{\olnot}[1]{\mkern 1.5mu\overline{\mkern-1.5mu#1\mkern-0.5mu}\mkern 0.5mu}

\newcommand{\restrict}[2]{{#1\!\!\upharpoonright_{#2}}}

\newcommand{\rup}[2]{\text{RUP}(#1, #2)}
\renewcommand{\rup}[2]{\mathit{RUP}(#1, #2)}

\newcommand{\dom}[0]{\textit{dom}}

\newcommand{\introduceterm}[1]{\emph{#1}}

\usepackage{xspace}

\newcommand{\redbasedstrengthening}{redundance-based strengthening\xspace}

\newcommand{\REDBASEDSTRENGTHENING}{Redundance-Based Strengthening\xspace}

\newcommand{\mapredundancy}{substitution redundancy\xspace}
\newcommand{\Mapredundancy}{Substitution redundancy\xspace}

\newcommand{\mapredundant}{substitution redundant\xspace}

\newcommand{\litell}{\ell}
\newcommand{\varx}{x}
\newcommand{\vary}{y}
\newcommand{\constrc}{C}
\newcommand{\constrd}{D}
\newcommand{\coeffa}{a}
\newcommand{\dega}{A}

\newcommand{\formf}{F}
\newcommand{\assmntrho}{\rho}
\newcommand{\assmntalpha}{\alpha}
\newcommand{\assmntbeta}{\beta}
\newcommand{\mapassmnt}{\omega}
\newcommand{\boolval}{b}
\newcommand{\numVars}{n}
\newcommand{\sumoneton}[1]{\sum_{#1=1}^{\numVars}}

\newcommand{\negc}[1]{\neg#1}

\newcommand{\Land}{\bigwedge}

\newcommand{\sumnodisplay}{{\textstyle \sum}}

\newcommand{\ceiling}[1]{\lceil #1 \rceil}

\newcommand{\floor}[1]{\lfloor #1 \rfloor}
\newcommand{\Floor}[1]{\bigl \lfloor #1 \bigr \rfloor}

\newcommand{\eqperiod}{\enspace .}
\newcommand{\eqcomma}{\enspace ,}

\newcommand{\mysubsection}[1]{\subsection{#1}}

\newcommand{\proofsystemformat}[1]{\textit{#1}}
\newcommand{\toolformat}[1]{\textit{#1}}

\newcommand{\rat}{\proofsystemformat{RAT}\xspace}
\newcommand{\drat}{\proofsystemformat{DRAT}\xspace}
\newcommand{\lrat}{\proofsystemformat{LRAT}\xspace}
\newcommand{\frat}{\proofsystemformat{FRAT}\xspace}
\newcommand{\grit}{\proofsystemformat{GRIT}\xspace}
\newcommand{\pbp}{\proofsystemformat{PBP}\xspace}
\newcommand{\tracecheck}{\proofsystemformat{TraceCheck}\xspace}
\newcommand{\rupname}{\proofsystemformat{RUP}\xspace}

\newcommand{\veripb}{\toolformat{VeriPB}\xspace}
\newcommand{\minisat}{\toolformat{MiniSat}\xspace}
\newcommand{\satforj}{\toolformat{Sat4j}\xspace}
\newcommand{\roundingsat}{\toolformat{RoundingSat}\xspace}
\newcommand{\kissat}{\toolformat{Kissat}\xspace}
\newcommand{\prtodrat}{\toolformat{PR2DRAT}\xspace}

\newcommand{\cryptominisat}{\toolformat{CryptoMiniSat}\xspace}
\newcommand{\slime}{\toolformat{SLIME}\xspace}
\newcommand{\cnfgen}{\toolformat{CNFgen}\xspace}
\newcommand{\drattrim}{\toolformat{DRAT-trim}\xspace}

\newcommand{\refsec}[1]{Section~\ref{#1}}

\newcommand{\reffig}[1]{Figure~\ref{#1}}

\ifthenelse
{\isundefined{\refeq}}
{\newcommand{\refeq}[1]{\eqref{#1}}}
{\renewcommand{\refeq}[1]{\eqref{#1}}}

\definecolor{mGreen}{rgb}{0,0.6,0}
\definecolor{mGray}{rgb}{0.5,0.5,0.5}
\definecolor{mPurple}{rgb}{0.58,0,0.82}
\definecolor{backgroundColour}{rgb}{0.95,0.95,0.92}

\newcommand{\Nplus}{\mathbb{N}_+}
\newcommand{\N}{\mathbb{N}_0}
\renewcommand{\Nplus}{\mathbb{N}^+}
\renewcommand{\N}{\mathbb{N}}

\providecommand{\coeffa}{a}
\providecommand{\coeffb}{b}

\providecommand{\consta}{A}
\providecommand{\constb}{B}

\providecommand{\litell}{\ell}

\providecommand{\prooflabel}[1]{\scriptsize{\textsf{#1}}}
\renewcommand{\prooflabel}[1]{\footnotesize{\textsf{#1}}}
\providecommand{\rightprooflabel}[1]{\ \prooflabel{#1}}

\providecommand{\pbconstra}{\textstyle \sum_i \coeffa_i \litell_i \geq \consta}
\providecommand{\pbconstrb}{\textstyle \sum_i \coeffb_i \litell_i \geq \constb}
\providecommand{\pbconstrlincomb}[7]%
    {\textstyle \sum_{#1}
      ({#6} {#2}_{#1} + {#7}{#4}_{#1}) \litell_{#1}
      \geq
      {#6} {#3} + {#7}{#5}}

\renewcommand{\set}[1]{\{ #1 \}}
\renewcommand{\Set}[1]{\bigl\{ #1 \bigr\}}

\newcommand{\setdescr}[3][\mid]{\set{ #2 #1 #3 }}
\newcommand{\Setdescr}[3][|]%
     {\ifthenelse{\equal{#1}{;}}%
     {\Set{ #2 \,;\, #3 }}
     {\ifthenelse{\equal{#1}{:}}%
     {\Set{ #2 \,:\, #3 }}
     {\twincommandJN{\bigl\{}{#2\,}{\bigl#1}{\bigr}{\,#3}{\bigr\}}}}}
\newcommand{\SETDESCR}[3][|]%
     {\twincommandJN{\left\{}{#2\,}{\left#1}{\right}{\,#3}{\right\}}}

\newcommand{\setsize}[1]{\lvert#1\rvert}

\newcommand{\acceptword}{pass\xspace}
\renewcommand{\acceptword}{accept\xspace}
\newcommand{\rejectword}{fail\xspace}
\renewcommand{\rejectword}{reject\xspace}

\newcommand{\commentbracketsmath}[1]{[ \, #1 \, ]}
\newcommand{\commentbrackettext}[1]{\commentbracketsmath{\text{#1}}}

\renewcommand{\olnot}[1]{\overline{#1}}

\newcommand{\varstrue}[1]{\mathcal{T}(#1)}
\newcommand{\varsfalse}[1]{\mathcal{F}(#1)}
\newcommand{\yprimesigned}{y'(b)}

\newcommand{\clnotass}[1]{C_{\lnot{#1}}}
\newcommand{\clnotassext}[3]{\clnotass{(#1 \, \cup \, \set{#2 \mapsto #3 })}}
\newcommand{\clnotassrhoy}[1][\rho]{\clnotassext{#1}{\yprimesigned}{0}}
\newcommand{\clnotassextxy}[3]%
   {\clnotass{(#1 \cup \set{#2 \mapsto #3, \, \yprimesigned \mapsto 0 })}}

\newcommand{\reifequiv}{\Leftrightarrow}

\providecommand{\pclit}{\mathit{lit}}
\renewcommand{\pclit}{\litell}
\providecommand{\pcreason}{\mathit{reason}}
\renewcommand{\pcreason}{\constrc_{\mathrm{reason}}}
\providecommand{\pclearned}{\constrc_{\mathrm{learned}}}

\providecommand{\pcformatvalue}[1]{\mathtt{#1}}
\providecommand{\pcnull}{\pcformatvalue{NULL}}
\providecommand{\pcsat}{\pcformatvalue{SATISFIABLE}}
\providecommand{\pcunsat}{\pcformatvalue{UNSATISFIABLE}}

\providecommand{\parityreasoningformat}[1]{\textbf{\textit{#1}}}
\providecommand{\parityreasoningformatdescription}{boldface italicized\xspace}

\usepackage{numname}
\newcounter{authorcount}
\newcommand{\newauthor}[3]{%
\newcounter{#1comment}
\expandafter\newcommand\csname #1comment\endcsname[1]{%
\ifdefined\DONOTINSERTCOMMENTS\relax\else%
\par\noindent
{
\scshape \footnotesize #2's comment
\stepcounter{#1comment}\csname the#1comment\endcsname}:
{\sffamily \itshape
\textcolor{#3}{##1}\par}
\fi}
\stepcounter{authorcount}
\expandafter\edef\csname #1ordinal\endcsname{\theauthorcount}
\expandafter\newcommand\csname theauthor#1\endcsname{%
the \ordinaltoname{\csname #1ordinal\endcsname} author\xspace}
\expandafter\newcommand\csname Theauthor#1\endcsname{%
The \ordinaltoname{\csname #1ordinal\endcsname} author\xspace}
}

\definecolor{airforceblue}{rgb}{0.36, 0.54, 0.66}
\definecolor{amethyst}{rgb}{0.6, 0.4, 0.8}
\definecolor{asparagus}{rgb}{0.53, 0.66, 0.42}
\definecolor{brass}{rgb}{0.71, 0.65, 0.26}
\definecolor{brown}{rgb}{0.59, 0.29, 0.0}
\definecolor{darkolivegreen}{rgb}{0.33, 0.42, 0.18}
\definecolor{darkorange}{rgb}{1.0, 0.55, 0.0}
\definecolor{darkcyan}{rgb}{0, 0.5, 0.5}

\newauthor{bb}{Bart}{darkolivegreen}
\newauthor{sg}{Stephan}{darkcyan}
\newauthor{cm}{Ciaran}{darkorange}
\newauthor{jn}{Jakob}{blue}

\ifthenelse{\boolean{publisherversion}}
{}
{
  \numberwithin{equation}{section}
}

\newcommand{\jnshortcite}[1]{\cite{#1}}
\newcommand{\jnciteA}[1]{\cite{#1}}
\newcommand{\jnciteNameRef}[2]{#2~\cite{#1}}
\newcommand{\jnciteinorby}[1]{in~\cite{#1}\xspace}

\newcommand{\jncitenameurl}[3]{#1~\cite{#2}\xspace}
\newcommand{\jncitenameurlpunctuation}[4]{#1~\cite{#2}#4\xspace}

\begin{document}

\pdfinfo{
/Title (Certifying Parity Reasoning Efficiently Using Pseudo-Boolean Proofs)
/Author (Stephan Gocht, Jakob Nordstrom)
/TemplateVersion (2021.1)
} %

\title{Certifying Parity Reasoning Efficiently \\ 
  Using Pseudo-Boolean Proofs%
  \thanks{This is the full-length version of the 
    conference paper~\cite{GN21CertifyingParity}
    presented at \emph{AAAI~'21}.}
}

\author {
      Stephan Gocht,\textsuperscript{\rm 1,2}
      Jakob Nordstr\"om,\textsuperscript{\rm 2,1} \\
      \textsuperscript{\rm 1} Lund University, Lund, Sweden \\
      \textsuperscript{\rm 2} University of Copenhagen, Copenhagen, Denmark \\
      stephan.gocht@cs.lth.se, jn@di.ku.dk
}
\date{\today}

\maketitle

\ifthenelse{\boolean{publisherversion}}
{}
{
  \thispagestyle{empty}

  \pagestyle{fancy}
  \fancyhead{}
  \fancyfoot{}
  \renewcommand{\headrulewidth}{0pt}
  \renewcommand{\footrulewidth}{0pt}

  \fancyhead[CE]{\slshape
    CERTIFYING PARITY REASONING EFFICIENTLY USING PSEUDO-BOOLEAN PROOFS}
  \fancyhead[CO]{\slshape \nouppercase{\leftmark}}
  \fancyfoot[C]{\thepage}

  \setlength{\headheight}{13.6pt}
}

\begin{abstract}
  The dramatic improvements in combinatorial optimization algorithms
  over the last decades have had a major impact in artificial
  intelligence, operations research, and beyond, but the output of
  current state-of-the-art solvers is often hard to verify and is
  sometimes wrong. For Boolean satisfiability (SAT) solvers proof
  logging has been introduced as a way to certify correctness, but the
  methods used seem hard to generalize to stronger paradigms. What is
  more, even for enhanced SAT techniques such as parity (XOR)
  reasoning, cardinality detection, and symmetry handling, it has
  remained beyond reach to design practically efficient proofs in the
  standard \drat format. In this work, we show how to instead use
  pseudo-Boolean inequalities with extension variables to concisely
  justify XOR reasoning. Our experimental evaluation of a SAT solver
  integration shows a dramatic decrease in proof logging and
  verification time compared to existing \drat methods. Since our
  method is a strict generalization of \drat, and readily lends itself
  to expressing also 0-1 programming and even constraint programming
  problems, we hope this work points the way towards a unified
  approach for 
  efficient machine-verifiable proofs for a rich class of
  combinatorial optimization paradigms.
\end{abstract}

\section{Introduction}
\label{sec:intro}

Since around the turn of the millennium, combinatorial optimization
has been successfully applied to solve an ever increasing range of
problems in
e.g., resource allocation, scheduling, logistics, and disaster 
management~\cite{HandbookCombOpt13},
and more recent applications in
biology, chemistry, and 
medicine
include, e.g.,  protein analysis and 
design~\jnshortcite{AABDGKOPST14ComputationalProteinDesign,MWB08CPSPtools}
and 
planning for
kidney transplants~\jnshortcite{MO12PairedAltruisticKidney,BKMPetal21Kidney}.
Yet other examples are government auctions generating billions of
dollars in revenue~\cite{LMS17Economics},
as well as allocation of education and work 
opportunities~\cite{Manlove16HospitalsResidents,MMT17AlmostStable}
and matching of
adoptive families with children~\jnshortcite{DGGKMP19MathematicalModels}.

As more and more such problems are dealt with using combinatorial
optimization solvers, an urgent question is whether we can trust that
the solutions computed by such algorithms are correct and complete.
The answer, unfortunately, is currently a clear ``no'':
State-of-the-art 
solvers
sometimes return ``solutions'' that do not satisfy the constraints
or erroneously claim optimality
of solutions~%
\jnshortcite{CKSW13Hybrid,AGJMN18Metamorphic,GSD19SolverCheck}.
This can be fatal for applications such as, e.g., chip design,
compiler optimization, and combinatorial auctions, where correctness
is absolutely crucial, not to speak about when human lives depend on
finding the best solutions.

Conventional software testing 
has
made little progress in
addressing this problem, and formal verification techniques cannot
handle the level of complexity of modern solvers.
Instead, the most successful approach to date has been that of
\emph{proof logging} in the Boolean
satisfiability (SAT) community,
where solvers are required to be
\emph{certifying}~\cite{MMNS11CertifyingAlgorithms}
in the sense that they output not only 
a result but also a simple, machine-verifiable proof that this result is
correct.

This does not certify the correctness of the solver itself, but it
does mean that if it ever produces an incorrect answer
(even if due to hardware errors), 
then this can be detected.  
Furthermore, such proofs can in principle be stored and audited later
by a third party using independently developed software.
A number of different proof logging formats such as
\rupname~\cite{GN03Verification},
\tracecheck~\cite{Biere06TraceCheckWithDate},
\drat~\cite{HHW13Trimming,HHW13Verifying,WHH14DRAT}, 
\grit~\cite{CMS17GRIT},
and
\lrat~\jnshortcite{CHHKS17LRAT}
have been developed, with
\drat now established as the standard in the 
\jncitenameurlpunctuation
{SAT competitions}
{SATcompetition}
{http://www.satcompetition.org}
{.}

A quite natural, and highly desirable, goal would be to extend these
proof logging techniques to stronger
combinatorial optimization     %
paradigms such as pseudo-Boolean (PB) optimization, MaxSAT solving,
mixed integer
linear   %
programming (MIP),
and 
constraint programming (CP),
but such attempts have had limited success.
Either the proofs require trusting in powerful and complicated rules
(as in, e.g.,~\cite{VS10ProofProducing}), defeating simplicity and
verifiability, or they have to justify such rules by long explanations,
leading to an exponential slow-down (see~\cite{GS19Certifying}).
In fact, even for SAT solvers a long-standing problem is that 
more advanced techniques  
for detecting and
reasoning with parity
constraints (a.k.a.\ exclusive or, or XOR, constraints),  
cardinality constraints, and symmetries
have remained out of reach for efficient proof logging.
Although in theory 
it might seem like
there should be no problems---the \drat proof system
is extremely powerful, and can in principle justify
such reasoning and much more with at most a polynomial amount of
work~\cite{SB06ExtendedResolution,HHW15SymmetryBreaking,PR16DRATforXOR}---in
practice the overhead seems completely prohibitive.
Thus, a key challenge on the road to efficient proof logging for 
more general combinatorial 
optimization 
solvers would seem to be to
design a method that can capture the full range of techniques used in
modern SAT solvers.

\mysubsection{Our Contribution}
In this work, we present a new, efficient proof logging method for
parity reasoning that is---perhaps somewhat surprisingly---based on
pseudo-Boolean reasoning with \mbox{$0$-$1$} 
integer
linear inequalities.
Though such inequalities might seem ill-suited to representing XOR
constraints, this can be done elegantly by introducing auxiliary so-called 
\emph{extension  variables}~\cite{DGP04GeneralizingBooleanSatisfiabilityI}. 
Using this observation, we strengthen the 
\jncitenameurl{\veripb tool}{VeriPB}{https://gitlab.com/MIAOresearch/VeriPB}
recently introduced \jnciteinorby{EGMN20Justifying}, which can be viewed as a
generalization to pseudo-Boolean proofs of 
\rupname~\cite{GN03Verification}.
Borrowing inspiration
from~\jnciteA{HKB17ShortProofs,BT19DRAT}, we develop stronger, but still
efficient, rules that can handle also extension variables, making
\veripb, in effect, into a strict generalization of \drat.

We have implemented our method 
for representing XOR constraints and performing Gaussian elimination
on such constraints
in a library with a
simple, clean interface for SAT solvers. As a proof of concept, we
have also integrated it in \minisat~\cite{ES03MiniSat}, which still
serves as the foundation of 
many
state-of-the-art SAT solvers.
Our library also provides \drat proof logging for XORs  
as described \jnciteinorby{PR16DRATforXOR}, but with some optimizations,
to allow for a comparative evaluation. Our experiments show
that the overhead for proof logging, 
the size of the produced proofs, and the time for verification 
all go down by orders of magnitude 
for our
pseudo-Boolean
method 
compared to
\drat.
Furthermore, the fact that PB reasoning forms the basis for solvers like
\satforj~\cite{LP10Sat4j}
and
\roundingsat~\cite{EN18RoundingSat}
means that our library can also empower such pseudo-Boolean
solvers 
to reason with parities.

Since cardinality constraints are just a special case of PB
constraints, it is clear that our method should suffice to justify
the cardinality reasoning used in SAT solvers. The method presented in
this paper is not sufficient for efficient proof logging of general
symmetry breaking,
but at least we can perform as efficiently for symmetry breaking as any
approach using \drat, since our proof system subsumes \drat.
More excitingly, the original \veripb tool has already been shown to be
capable of efficiently justifying a number of constraint programming
techniques~\jnshortcite{EGMN20Justifying,GMN20SubgraphIso,GMMNPT20CertifyingSolvers}.
Our optimistic interpretation
is that pseudo-Boolean reasoning
with extension variables shows great potential as a unified method
of proof logging for SAT solving, pseudo-Boolean optimization,
MaxSAT solving,
constraint programming, and maybe even mixed integer programming.

\subsection{Subsequent Developments}

The last couple of years have witnessed quite significant developments in
proof logging.
Since the conference version of this paper appeared,
our pseudo-Boolean proof logging method has been extended further to
deal with fully general symmetry breaking in SAT
solving~\jnshortcite{BGMN22Dominance}, 
and also to support pseudo-Boolean solving using SAT
solvers~\jnshortcite{GMNO22CertifiedCNFencodingPB}.
Furthermore, there have been promising preliminary results on
providing proof logging for  
MaxSAT solvers~\cite{VWB22QMaxSATpb}
and constraint programming solvers~\cite{GMN22AuditableCP}.

The \drat proof logging method has recently been extended to
\frat~\cite{BCH21Flexible}, which allows to integrate different forms
of reasoning.
Proof logging using binary decision diagrams
(BDDs)~\cite{Bryant22TBUDDY},
generating proofs in all of the
\drat, \lrat, and \frat formats,
has also been developed for pseudo-Boolean
reasoning~\cite{BBH22ClausalProofs} 
and parity reasoning~\cite{SB22Combining}.
Further evaluation will be needed to decide whether such clausal proof
logging methods can be truly competitive with pseudo-Boolean proof
logging.

\subsection{Organization of This Paper}
After some brief background in
\refsec{sec:prelims}, we introduce the key technical 
notions needed
for our new proof logging rules
in
\refsec{sec:mapping-redundancy}
and show how they can be used to justify parity reasoning in
\refsec{sec:proof-logging} with a worked out example in \refsec{sec:example}.
We present an experimental evaluation in
\refsec{sec:evaluation}
and provide some concluding remarks in
\refsec{sec:conclusion}.

\section{Preliminaries}
\label{sec:prelims}

Let us start by quickly reviewing the required material on
pseudo-Boolean reasoning, referring the reader to, e.g.,
\jnciteA{BN21ProofCplxSAT} for more context.
A few pieces of standard notation are that we write
$
\N = \set{0, 1, 2, \ldots}
$
to denote non-negative integers and
$
\Nplus = \N \setminus \set{0}
$ 
to denote positive integers.
For $\numVars \in \Nplus$,
we write $[\numVars] = \set{1, 2, \ldots, \numVars}$
to denote the set consisting of the first~$\numVars$ positive integers.

A \emph{literal}~$\litell$ over a Boolean variable~$\varx$
is $\varx$ itself or its negation
\mbox{$\olnot{\varx} = 1 - \varx$},
where variables take values~$0$ (false) or $1$~(true).
For notational convenience, we define
$\olnot{\olnot{\varx}} = \varx$.
A \emph{pseudo-Boolean (PB) constraint}~$\constrc$
over
literals $\litell_1, \dots, \litell_\numVars$
is a \mbox{$0$-$1$} linear inequality
\begin{equation}
  \label{eq:normalized}
  \sumoneton{i} \coeffa_i \litell_i \geq \dega
  \eqcomma
\end{equation}
which without loss of generality we always assume to be in
\emph{normalized form}; i.e.,
all literals $\litell_i$ are over distinct variables and
the coefficients~$\coeffa_i$ and  the
\emph{degree (of falsity)}~$\dega$
are non-negative integers. 
Conversion to normalized form can be performed efficiently by using equalities
\mbox{$\olnot{\varx} = 1 - \varx$} to rewrite the left-hand side of
any inequality as a positive linear combination of literals, and so in
what follows we will consider any pseudo-Boolean constraint and its
normalized form to be one and the same constraint.
We will use equality
\begin{subequations}
\begin{align}
  \label{eq:equality}
  \sumoneton{i} \coeffa_i \litell_i &= \dega \\
\intertext{as syntactic sugar for the pair of inequalities}
  \label{eq:equality-ineq-1}
  \sumoneton{i} \coeffa_i \litell_i &\geq \dega \\
  \label{eq:equality-ineq-2}
  \sumoneton{i} -\coeffa_i \litell_i &\geq -\dega
\end{align}
\end{subequations}
(but rewritten in normalized form)
and the \emph{negation}~$\negc{\constrc}$ of~\refeq{eq:normalized} is
(the normalized form of)
\begin{equation}
  \label{eq:negation}
  \sumoneton{i} -\coeffa_i \litell_i \geq -\dega + 1
  \eqperiod
\end{equation}
A \emph{pseudo-Boolean formula} is a
conjunction
$\formf = \Land_{j=1}^{m} \constrc_j$
of pseudo-Boolean constraints.
Note that a 
\emph{clause}
$\litell_1 \lor \cdots \lor \litell_k$
is equivalent to the 
constraint
\mbox{$\litell_1 + \cdots + \litell_k \geq 1$},
so formulas in
\emph{conjunctive normal form (CNF)}
are special cases of 
pseudo-Boolean formulas.

A \emph{(partial) assignment} is a (partial) function from variables
to $\set{0,1}$ and a substitution is a (partial) function from
variables to literals or $\set{0,1}$. For an assignment or
substitution~$\assmntrho{}$ 
we will use the convention $\assmntrho(\varx{}) = \varx{}$ 
for $\varx$  not in the domain of~$\assmntrho$, 
denoted $\varx{} \not \in \dom(\assmntrho{})$, and define $\assmntrho(\olnot{\varx})= 1 -
\assmntrho(\varx)$. We also write $\varx \mapsto \boolval$ instead of
$\assmntrho(\varx) = \boolval$, where $\boolval$ 
denotes
$0$, $1$, or a literal, when $\assmntrho$ is clear from context or
is immaterial.
Applying $\assmntrho$ to a 
pseudo-Boolean %
constraint~$\constrc$ as in~\refeq{eq:normalized},
denoted $\restrict{\constrc}{\assmntrho}$, yields the
constraint obtained by
substituting values for all assigned variables,
shifting constants to the right-hand side, and adjusting the
degree appropriately, i.e.,
\begin{equation}
  \label{eq:restricted-constraint}
  \restrict{\constrc}{\assmntrho} \, =\, 
  \sumnodisplay_{i} \coeffa_i \assmntrho(\litell_i) 
  \geq
  \dega
\end{equation}
with appropriate normalization, and for a formula~$\formf$ we define
$
\restrict{\formf}{\assmntrho}
=
\Land_j  \restrict{\constrc_j}{\assmntrho}
$.
The
normalized  %
constraint~$\constrc$
is \emph{satisfied} by~$\assmntrho$
if
$\sum_{\assmntrho(\litell_i) = 1} \coeffa_i \geq \dega$
(or, equivalently, if the restricted
constraint~\refeq{eq:restricted-constraint}
has a non-positive degree and is thus trivial).
A PB formula is satisfied by~$\assmntrho$ if all constraints in it are,
in which case it is \emph{satisfiable}.
If there is no satisfying assignment, the formula is
\emph{unsatisfiable}. Two formulas are
\emph{equisatisfiable} if they are both
satisfiable or both unsatisfiable.

The \introduceterm{cutting planes} proof system as defined
\jnciteinorby{CCT87ComplexityCP} is a method for iteratively deriving new
constraints~$\constrc$ implied by a
pseudo-Boolean  %
formula~$\formf$.
Cutting planes contains
rules for
\introduceterm{literal axioms}
\begin{equation}
  \label{eq:cp-rule-literal}
  \AxiomC{\rule{0pt}{8pt}}
  \UnaryInfC{$\ \litell_i \geq 0 \ $}
  \DisplayProof   \ ,
\end{equation}
and \introduceterm{linear combinations}
\begin{equation}
  \label{eq:cp-rule-lincomb}
    \AxiomC{$\pbconstra    $}
    \AxiomC{$\pbconstrb$}
    \BinaryInfC{$
    \pbconstrlincomb{i}{\coeffa}{\consta}{\coeffb}{\constb}{c_A}{c_B}
    $}
    \DisplayProof
    \ \ \ 
    \commentbracketsmath{c_A, c_B \in \N}
    \eqperiod
\end{equation}
For notational convenience, in this paper we will sometimes use linear
combinations of equalities as in~\eqref{eq:equality}, which is just a
shorthand for taking pairwise linear combinations of inequalities of the
form~\eqref{eq:equality-ineq-1} 
and~\eqref{eq:equality-ineq-2},
respectively.
There is also a rule
for \introduceterm{division}
\begin{equation}
  \label{eq:cp-rule-div}
    \AxiomC{$
    \textstyle \sum_i \coeffa_i \litell_i \geq \consta
    $}
    \UnaryInfC{$
    \textstyle \sum_i
    \ceiling{\coeffa_i / c } \litell_i
    \geq
    \ceiling{\consta / c}
    $}
    \DisplayProof
    \ \ \ 
    \commentbracketsmath{c \in \Nplus}
\end{equation}
(where we note that the soundness of this rule depends on that the
pseudo-Boolean constraint is written in normalized form).
As a toy example, the derivation
\begin{equation}
\label{eq:cp-toy-example}
  \AxiomC{$
    6x + 2y + 3z \geq 5
    $}
  \AxiomC{$
    x + 2y + w \geq 1
    $}
  \RightLabel{\rightprooflabel{%
        Linear combination ($c_A = 1$, $c_B = 2)$%
    }}
  \BinaryInfC{%
    $8x + 6y + 3z + 2w \geq 7$
  }
  \RightLabel{\rightprooflabel{Division ($c = 3$)}}
  \UnaryInfC{$3x + 2y + z + w \geq 3$}
  \DisplayProof
\end{equation}
illustrates how these rules can be combined to obtain new constraints.

The proof system that we use for the proof logging in \veripb
also supports additional rules such as
the \introduceterm{saturation rule}
\begin{equation}
  \label{eq:cp-rule-saturation}
  \AxiomC{$\sum_i \coeffa_i \litell_i \geq \consta$}
  \UnaryInfC{$
    \sum_i \min(\coeffa_i,\consta) \cdot \litell_i \geq \consta
    $}
  \DisplayProof \eqcomma
\end{equation}
which is not part of the cutting planes proof system
defined \jnciteinorby{CCT87ComplexityCP} but was
introduced in the context of pseudo-Boolean
solving
\jnciteinorby{CK05FastPseudoBoolean}. For example, from
the constraint 
$8x + 6y + 3z + 2w \geq 7$ 
in the example above 
it is possible to derive $7x + 6y + 3z + 2w \geq 7$ via
saturation. While this might not be clear from a small example like
this, the division and saturation rules are incomparable in
strength~\cite{GNY19DivisionSaturation}.

For 
pseudo-Boolean   %
formulas $\formf$, $\formf'$ and constraints $\constrc$,
$\constrc'$, we say that
$\formf$ \emph{implies} or \emph{models}~$\constrc$,
\mbox{denoted $\formf \models \constrc$}, if
any  assignment satisfying~$\formf$
must also satisfy~$\constrc$,
and we write $\formf \models \formf'$ if
$\formf \models \constrc'$  for all $\constrc' \in \formf'$.
It is not hard to see that any collection of constraints~$\formf'$
derived (iteratively) from~$\formf$ by
cutting planes are implied in this sense,
and so it holds that
$\formf$
and
$\formf \land \formf'$
are equisatisfiable.
A particularly simple type of implication is when a 
constraint~$\constrc'$ can be derived from some other constraint~$\constrc$ using
only addition of literal axioms as in~\eqref{eq:cp-rule-literal}. 
When this is is the case, we will say that $\constrc'$ is 
\emph{implied syntactically}
by~$\constrc$.

A constraint~$\constrc$ is said to  \emph{unit propagate}
the literal~$\litell$ under~$\assmntrho$
if $\restrict{\constrc}{\assmntrho}$
cannot be satisfied unless
$\litell \mapsto 1$.
During \emph{unit propagation} on $\formf$ under~$\assmntrho$,
we extend~$\assmntrho$ 
iteratively by any propagated literals
$\litell \mapsto 1$ until an assignment~$\assmntrho'$ is reached under
which no constraint
$\constrc \in \formf$
is propagating, or
under which some constraint~$\constrc$ propagates a literal that has
already been assigned to the opposite value.
The latter scenario is referred to as a
\emph{conflict}, since $\assmntrho'$ \emph{violates} the
constraint~$\constrc$ in this case, and $\assmntrho'$ is called a
\emph{conflicting} assignment.

Using the generalization of~\cite{GN03Verification}
\jnciteinorby{EGMN20Justifying},
we say that  $\formf$  implies~$\constrc$ by
\emph{reverse unit propagation (RUP)}, and write
$\rup{\formf}{\constrc}$,
if
$\formf \land \negc{\constrc}$
unit propagates to conflict under the empty assignment.
It is not hard to see that
$\rup{\formf}{\constrc}$
implies
$\formf \models \constrc$,
but the opposite direction is not necessarily true.
Cutting planes as defined above is
not only \emph{sound} in the sense that it can only derived implied
constraints, but it is also \emph{implicationally complete},
which means that if a pseudo-Boolean formula~$\formf$
implies a constraint~$\constrc$, then there is also a cutting planes
derivation of~$\constrc$ from~$\formf$. This holds, in particular, if
$\constrc$ is RUP with respect to~$\formf$. However, it might not
always be obvious how to construct such a derivation, and therefore 
we can add a derivation rule for adding RUP constraints as a
convenient shorthand.
An important special case of completeness is that if a set of
pseudo-Boolean constraints~$\formf$ is unsatisfiable, 
then there exists a cutting planes derivation of the contradiction
$0 \geq 1$ from~$\formf$, which we refer to as a 
\emph{proof of unsatisfiability}, or \emph{refutation}, of~$\formf$.

\section{\REDBASEDSTRENGTHENING{}}
\label{sec:mapping-redundancy}

In order to provide proof logging for parity reasoning, we need the
ability not only to perform cutting planes reasoning, but also to
introduce
\emph{fresh variables} not occurring
in the formula~$\formf$ under consideration. 
In particular, we want to be able to use a fresh variable~$\vary$
to encode the \emph{reification} of a constraint
$\sum_i \coeffa_i \litell_i \geq \dega$,
i.e., that $\vary$ is true if and only if the constraint is satisfied.  
We will use the shorthand
\begin{equation}
  \label{eq:reification-notation}
  \vary \reifequiv
  \sumnodisplay_i \coeffa_i \litell_i \geq \dega
\end{equation}
to denote
the two constraints
\begin{subequations}
  \begin{align}
    \label{eq:reification-encoding-1}
    \dega \olnot{\vary} + 
    \sumnodisplay_i \coeffa_i \litell_i &\geq \dega
    \\
    \label{eq:reification-encoding-2}
    \bigl( -\dega \!+\! 1  \!+\! \sumnodisplay_i \coeffa_i \bigr) \cdot \vary + 
    \sumnodisplay_i \coeffa_i \olnot{\litell}_i &\geq
    -\dega + 1  + \sumnodisplay_i \coeffa_i
  \end{align}
\end{subequations}
enforcing this condition
(which is the case under the the assumption that the constraint 
$\sum_i \coeffa_i \litell_i \geq \dega$
is written in normalized form).
By way of a concrete example, the reification of the constraint
\begin{equation}
  \label{eq:reified-constraint-example}
  \varx_1 + \varx_2 + \varx_3 \geq 2    
\end{equation}
using $\vary$ is encoded as
\begin{subequations}
  \begin{align}
    \label{eq:reification-example-encoding-1}
    2 \olnot{\vary} + 
    \varx_1 + \varx_2 + \varx_3 &\geq 2
    \\
    \label{eq:reification-example-encoding-2}
    2 \vary + 
    \olnot{\varx}_1 + \olnot{\varx}_2 + \olnot{\varx}_3 &\geq 2    
  \end{align}
\end{subequations}
in pseudo-Boolean form.  Note that introducing such constraints
maintains equisatisfiability provided that 
the \emph{reification variable}~$\vary$ 
does not appear in any other constraint, 
since any assignment to the literals~$\litell_i$ will satisfy
either~\refeq{eq:reification-encoding-1}
or~\refeq{eq:reification-encoding-2}, 
which allows us to assign~$y$ so that the other constraint is also satisfied.

More generally, it would be convenient to allow the
``derivation'' of any constraint~$\constrc$ from~$\formf$ such that
$\formf$ and $\formf \land \constrc$ are equisatisfiable---in which
case we say that
$\constrc$ is \emph{redundant with respect to~$\formf$}---regardless
of whether $\formf \models \constrc$ holds or not.
A moment of thought reveals that such a completely generic rule would
be too good to be true---for any unsatisfiable formula~$\formf$ we would then
be able to ``derive'' contradiction (say, $0 \geq 1$) in just one
step, 
and such derivations would be hard to check for correctness.
What we need, therefore, is a sufficient criterion for redundancy of
pseudo-Boolean 
constraints that is simple to verify.
To this end, we generalize the characterization of redundancy
\jnciteinorby{HKB17ShortProofs,BT19DRAT} from CNF formulas to 
pseudo-Boolean 
formulas  as follows.

\begin{proposition}[\Mapredundancy{}]
  \label{prop:redundant}
  A 
  pseudo-Boolean   %
  constraint~$\constrc$   is redundant with respect to 
  the formula~$\formf$ if and only if  there is a
  substitution $\mapassmnt$, called 
  a
  \emph{witness},
  for which
  it holds that
  \begin{equation*}
    \label{eq:redundant}
    \formf \land \negc{\constrc}
    \models
    \restrict{(\formf \land \constrc)}{\mapassmnt}
    \eqperiod
  \end{equation*}
\end{proposition}

\begin{proof}
($\Rightarrow$)
Suppose
$\constrc$ is redundant. If $\formf$ is unsatisfiable,
then for any constraint $\constrc'$ it vacuously holds that 
\mbox{$\formf \models \constrc'$}. 
Hence, any 
substitution $\mapassmnt$
fulfils the condition. If $\formf$ is satisfiable,  then 
$\formf \land \constrc$ must also be satisfiable as $\constrc$ is
redundant by assumption. If we choose $\mapassmnt$ to be a
satisfying assignment for $\formf \land \constrc$, 
the implication in the proposition again vacuously holds since
$\restrict{(\formf \land \constrc{})}{\mapassmnt}$ is fixed to true.

($\Leftarrow$)
Suppose now that $\mapassmnt$ is such that
\mbox{$
\formf \land \negc{\constrc}
\models
\restrict{(\formf \land \constrc)}{\mapassmnt}
$}.
If $\formf$ is unsatisfiable, then every constraint is
redundant and there is nothing to check. Otherwise, let
$\assmntalpha$ be a (total) satisfying assignment for $\formf$. 
If $\assmntalpha$ also satisfies~$\constrc$, then 
clearly
the constraint is redundant. Now consider the case that
$\assmntalpha$ does not satisfy $C$. If so, $\assmntalpha$ must
satisfy $\negc{\constrc}$ and hence, by the assumed implication,
also~$\restrict{(\formf \land \constrc{})}{\mapassmnt}$.
But then the assignment~$\assmntbeta$ defined by
\begin{equation}
  \assmntbeta(\varx) = 
  \begin{cases}
    \assmntalpha(\varx) & \text{if $\varx \not \in \dom(\mapassmnt)$,} \\
    \assmntalpha(\mapassmnt(\varx)) & \text{otherwise,}
  \end{cases}
\end{equation}
satisfies both~$\constrc$
and~$\formf$
(since
\mbox{$
\restrict{(\formf\land \constrc{})}{\assmntbeta}
=
\restrict{(\restrict{(\formf\land \constrc{})}{\mapassmnt})}{\assmntalpha}
$}
by construction),
so 
$\formf \land \constrc$
is satisfiable.
\end{proof}

We remark that this proof does not make use of that we are operating
with a pseudo-Boolean constraint~$\constrc$---we only need that 
the negation~$\negc{\constrc}$ is easy to represent in the same
formalism. Thus, the argument generalizes to other types of constraints
with this property
(such as, for instance, polynomial equations over finite fields when evaluated
on  Boolean inputs~$\set{0,1}^n$, as in the
\emph{polynomial calculus} proof 
system~\cite{CEI96Groebner,ABRW02SpaceComplexity}
formalizing Gröbner basis computations).

\newcommand{\clow}{\constrc_{\ref{eq:reification-example-encoding-2}}}
\newcommand{\chigh}{\constrc_{\ref{eq:reification-example-encoding-1}}}

Let us
return to our example reification of the constraint
in~\refeq{eq:reified-constraint-example} and show how this can be
derived
using \mapredundancy.  Let us write $\chigh$ for the constraint
in~\eqref{eq:reification-example-encoding-1} and $\clow$ 
for~\eqref{eq:reification-example-encoding-2},
where $\vary$ is fresh with respect to the current formula~$\formf$.
To show that $\clow$ is \mapredundant with respect to~$\formf$
we choose 
the
witness
$\mapassmnt = \set{\vary \mapsto 1}$,
which clearly satisfies~$\clow$.
Since $\vary$ does not appear in~$\formf$ we have
$
\restrict{\formf}{\mapassmnt} = \formf
$,
and so the implication
$
\formf \land \negc{ \clow }
\models \restrict{(\formf\land \clow)}{\mapassmnt}
$
vacuously holds. 
Showing that 
$\chigh$ is \mapredundant with respect to $\formf \land \clow$
is a bit more interesting.
For this we choose
$\mapassmnt = \set{\vary \mapsto 0}$,
which satisfies~$\chigh$
and again leaves~$\formf$ unchanged.  
Thus, the only implication for which we need to do some work is
$
\formf \land \clow \land \negc{ \chigh }
\models 
\restrict{\clow}{\mapassmnt}
$. 
The negation of $\chigh$ is
\begin{subequations}
\begin{equation}
  \label{eq:neg-C}
  -2 \olnot{\vary}  
  -\varx_1 - \varx_2 - \varx_3 \geq -1
\eqcomma
\end{equation}
or, converted to normalized form,
\begin{equation}
  \label{eq:neg-C-normalized}
  2 \vary + \olnot{\varx}_1 + \olnot{\varx}_2 + \olnot{\varx}_3 \geq  4
\end{equation}
\end{subequations}
using the rewriting rule
$\litell = 1  - \olnot{\litell}$.
Adding the literal axiom
$\olnot{\vary} \geq 0$
twice to $\negc{\chigh}$, and using rewriting again to cancel
\mbox{$\vary +  \olnot{\vary} = 1$}, 
we obtain
\begin{equation}
\olnot{\varx}_1 + \olnot{\varx}_2 + \olnot{\varx}_3 \geq 2 \eqcomma
\end{equation}
which is 
$\restrict{\clow}{\mapassmnt}$. 
Hence,
$\restrict{\clow}{\mapassmnt}$
can be derived from~$\negc{ \chigh }$
by just adding literal axioms---or, in the terminology introduced in
the preliminaries, $\negc{ \chigh }$~syntactically
implies~$\restrict{\clow}{\mapassmnt}$---and so
it certainly holds that   %
$\formf \land \clow \land \neg \chigh 
\models 
\restrict{\clow}{\mapassmnt}$.
This completes the proof that
$\chigh$ is redundant with respect to~${\formf\land\clow}$.

In our proof system for pseudo-Boolean proof logging, we will include
a \emph{\redbasedstrengthening{}}%
\footnote{In the conference version~\cite{GN21CertifyingParity} of
  this paper, this rule was called \emph{substitution redundancy}.
  However, since then an additional rule using witness substitutions
  has been introduced \jnciteinorby{BGMN22Dominance}, and we follow the
  terminology in this later paper to adhere to a consistent naming scheme.
} 
rule that allows to derive
constraints that satisfy the condition in
Proposition~\ref{prop:redundant}.
In order to do so, we need to discuss how the implication in this
\mapredundancy condition is to be verified.
Whenever this rule is used, the user needs to explicitly specify a 
witness~$\mapassmnt$, but this is not enough.
Arbitrary implication checks are as hard to verify as determining
satisfiability of a formula, and hence 
some kind of efficiently verifiable certificate
that the implication indeed holds 
is necessary 
to be able to validate the proof.
One way of providing
such %
a certificate
is to exhibit a cutting planes derivation establishing the validity of
the implication, as in the example just presented. A more convenient
alternative from a proof logging point of view is to follow the lead
of \drat and 
allow adding constraints without proof if 
the
implication can be verified automatically, e.g., using 
reverse
unit propagation. We
describe our pseudo-Boolean version of this 
automatic verification
method in
Algorithm~\ref{alg:PBRAT}. It is easy to see that the condition in
Proposition~\ref{prop:redundant}
is satisfied if Algorithm~\ref{alg:PBRAT}
issues a positive verdict:
If the algorithm 
accepts
because of $\rup{\formf}{\constrc}$,
then $\formf \models \constrc$ and it is 
in order
to add the constraint.
Otherwise, the algorithm 
will reject unless
for all constraints $\constrd$
in $\restrict{(\formf \land \constrc{})}{\mapassmnt}$, i.e., all
constraints on the right hand side of the implication 
in Proposition~\ref{prop:redundant},
it holds that
(a)~$\constrd \in \formf$, 
(b)~$\negc{ \constrc}$ implies $\constrd$ syntactically, or
(c)~$\rup{\formf \land \neg \constrc}{\constrd}$ evaluates to true. 
In all three cases it
follows that $\formf \land \constrc{} \models \constrd$, as desired.

We remark that this
algorithm is very similar to what is used 
for checking~\rat clauses in \drat proof verification, 
except that our unit propagation is on PB
constraints rather than clauses and that we need 
the additional 
syntactic check on
line~\ref{alg:PBRAT_implication_check} in Algorithm~\ref{alg:PBRAT}. 
To see why this 
extra step
is necessary, 
note that if we used only unit propagation, then we would fail to
certify the correctness 
of our example above.
Assuming for simplicity that $\formf = \emptyset$,
if we try to verify
$
\clow \land \negc{ \chigh }
\models \restrict{\clow}{\mapassmnt}
$
by reverse unit propagation we get the constraints
\begin{subequations}
\begin{align}
  \label{eq:failed-rup}
  2 \vary + 
  \olnot{\varx}_1 + \olnot{\varx}_2 + \olnot{\varx}_3 &\geq 2      
  &&
     \commentbrackettext{$\clow$ in~\eqref{eq:reification-example-encoding-2}}
  \\
  2 \vary + \olnot{\varx}_1 + \olnot{\varx}_2 + \olnot{\varx}_3 &\geq 4
  &&
     \commentbrackettext{$\negc{\chigh}$ in~\refeq{eq:neg-C-normalized}}
  \\
  \varx_1 + \varx_2 + \varx_3 &\geq 2      
  &&
     \commentbrackettext{negation of desired RUP constraint
     $\negc{(\restrict{\clow}{\mapassmnt})}$}
\end{align}
\end{subequations}
and although visual inspection shows that this collection of
constraints is inconsistent, since it requires
a majority of the variables
$\set{\varx_1, \varx_2, \varx_3}$
to be true and false at the same time, 
unit propagation  is too myopic to see this contradiction
and only yields 
$\vary \mapsto 1$.
Thanks to the fact that we instead use the stronger checks in
Algorithm~\ref{alg:PBRAT},  
we can automatically detect that the implication 
$
\clow \land \negc{ \chigh }
\models \restrict{\clow}{\mapassmnt}
$
holds.
This means that to introduce extension variables encoding reifications
$\vary \reifequiv \constrc$,
we do not need to do anything more than just specifying witness
assignments to the new variable as in the example above
for constraints
\refeq{eq:reification-example-encoding-1}
and~\refeq{eq:reification-example-encoding-2}.
For completeness, we write out the details in the general case for the
constraints \refeq{eq:reification-encoding-1}
and~\refeq{eq:reification-encoding-2} in the next proposition.

\providecommand{\pcformatvalue}[1]{\mathtt{#1}}
\providecommand{\pcaccept}{\pcformatvalue{ACCEPT}}
\providecommand{\pcreject}{\pcformatvalue{REJECT}}

\renewcommand{\acceptword}{$\pcaccept$}
\renewcommand{\rejectword}{$\pcreject$}

\begin{algorithm}[t]
  \begin{algorithmic}[1]
    \Procedure{RedundancyCheck}{$\formf, \constrc, \mapassmnt$}
        \Comment $\constrc, \mapassmnt$ are given in the
proof log 
    \If{$\rup{\formf}{\constrc}$}
    \Return \acceptword
    \EndIf
    \For{$\constrd \in \restrict{(\formf \land \constrc{})}{\mapassmnt}$}
    \If{\textbf{not} 
      $
      \bigl(    %
      \constrd \in \formf$ \textbf{or}
      $\negc{ \constrc}$ implies $\constrd$ syntactically
      \label{alg:PBRAT_implication_check} 
      \textbf{or}
      $\rup{\formf \land \neg \constrc}{\constrd}
      \bigr)    %
      $
    }
    \State \Return \rejectword
    \EndIf
    \EndFor
    \State\Return \acceptword
    \EndProcedure
  \end{algorithmic}
  \caption{Automatically checking \mapredundancy{}
    of $\constrc$ with respect to~$\formf$
  }
  \label{alg:PBRAT}
\end{algorithm}

\begin{proposition}
  \label{prop:reification-derivation}
  Let $\formf$ be a 
  pseudo-Boolean  %
  formula and $\constrc$ be a
  pseudo-Boolean   %
  constraint, 
  and
  suppose
  $\vary$ is a fresh variable that does not appear
  in~$\formf$ or~$\constrc$.
  Then the constraints
  \refeq{eq:reification-encoding-1}
  and~\refeq{eq:reification-encoding-2}
  encoding
  $\vary \reifequiv \constrc$
  can both be derived and added  to~$\formf$ 
  by the \redbasedstrengthening rule using
  Algorithm~\ref{alg:PBRAT}
  to verify the \mapredundancy conditions.
\end{proposition}

\renewcommand{\clow}{\constrc_{\ref{eq:reification-encoding-2}}}
\renewcommand{\chigh}{\constrc_{\ref{eq:reification-encoding-1}}}

\begin{proof}
Let us write $\chigh$ for the constraint
in~\eqref{eq:reification-encoding-1} and $\clow$
for~\eqref{eq:reification-encoding-2}.
To show that $\clow$ is \mapredundant with respect to~$\formf$ we
choose the witness $\mapassmnt = \set{\vary \mapsto 1}$, which clearly
satisfies~$\clow$.
Since the negation of a satisfied constraint is contradiction, 
this means, technically speaking, that 
$\restrict{\clow}{\mapassmnt}$
is RUP with respect to~$\formf$.
Since $\vary$ does not appear
in~$\formf$ we have $\restrict{\constrd}{\mapassmnt} = \constrd$ for all
$\constrd \in \formf$, which means that all constraints pass the check
on  Line~\ref{alg:PBRAT_implication_check} in  Algorithm~\ref{alg:PBRAT}.

As in our example above,
showing that $\chigh$ is \mapredundant with respect to $\formf \land
\clow$ 
requires slightly more work.
Here we choose the witness
$\mapassmnt = \set{\vary \mapsto 0}$, 
which satisfies~$\chigh$ and again
leaves~$\formf$ unchanged, which means that 
implication checks for  constraints in~$\formf$ are again vacuous.
The only constraint left to check is
$\restrict{\clow}{\mapassmnt}$, which is implied syntactically by
$\neg \chigh$, as we will see next. The negation of $\chigh$ is
\begin{subequations}
\begin{align}
    \dega \olnot{\vary} + 
    \sumnodisplay_i \coeffa_i \litell_i &\leq \dega - 1
\\
\shortintertext{or}
    -\dega \olnot{\vary} +
    \sumnodisplay_i -\coeffa_i \litell_i &\geq -\dega + 1
\eqcomma
\\
\intertext{which in normalized form becomes}
    \dega \vary +
    \sumnodisplay_i \coeffa_i \olnot{\litell}_i 
    &\geq 
    1 + \sumnodisplay_i \coeffa_i
\end{align}
\end{subequations}
(rewriting using the equality
$\litell = 1  - \olnot{\litell}$
to obtain a positive linear combination of literals on the left-hand
side of the inequality).
Adding $\dega$ times  the literal axiom
$\olnot{\vary} \geq 0$
to $\negc{\chigh}$ and applying cancellation
\mbox{$\vary +  \olnot{\vary} = 1$},
we obtain
\begin{equation}
\sumnodisplay_i \coeffa_i \olnot{\litell}_i \geq -\dega + 1 + \sumnodisplay_i \coeffa_i \eqcomma
\end{equation}
which is $\restrict{\clow}{\mapassmnt}$. Hence, $\negc{ \chigh }$
syntactically implies~$\restrict{\clow}{\mapassmnt}$, and so the 
condition on Line~\ref{alg:PBRAT_implication_check} is satisfied. 

This concludes the proof that the conditions required to derive the 
constraints~$\clow$ and~$\chigh$ by \redbasedstrengthening 
can be checked efficiently by Algorithm~\ref{alg:PBRAT}
regardless of what the pseudo-Boolean formula~$\formf$ is.
\end{proof}

\newcommand{\detectXor}{detectParities}
\newcommand{\propagate}{propagate}
\newcommand{\propagateXor}{gaussianElimination}
\newcommand{\nextDecision}{nextDecsion}
\newcommand{\analyse}{analyse}

\section{Proof Logging for XOR Constraints}
\label{sec:proof-logging}

We now proceed to explain how the cutting planes proof system in
\refsec{sec:prelims}
extended with the \redbasedstrengthening rule in
\refsec{sec:mapping-redundancy}
can be used to certify the correctness of parity reasoning.

An \emph{XOR} or \emph{parity constraint}, i.e., an equality modulo 2,
over $k$
Boolean
variables is written as
\begin{equation}
  \label{eq:xor}
  x_1 \oplus x_2 \oplus \dots \oplus x_\s{k} = b 
\end{equation}
for
$b \in \set{0,1}$. 
Note that we can assume that there is no
parity constraint with a negated variable $\olnot{\varx{}}$, 
because we can always substitute
$\olnot{\varx{}} = \varx{} \oplus 1$.
Systems of XOR constraints can be handled in a solver through Gaussian
elimination
\jnshortcite{SNC09SAT4Crypto,HR12SatMeetsGauss,LJN12CompleteParity} or
conflict analysis~\jnshortcite{LJN12ConflictXor}. 
In this paper we will focus on the integration of Gaussian elimination
into conflict-driven clause learning
(CDCL)~\jnshortcite{BS97UsingCSP,MS99Grasp,MMZZM01Engineering}, and so we
start by a quick review of how the CDCL main loop works and how parity
reasoning is included. The reader can consult the pseudocode in
Algorithm~\ref{alg:CDCL} to complement the description below.

\providecommand{\cdcldatabase}{\mathcal{D}}
\providecommand{\cdcltrail}{\mathit{trail}}

\begin{algorithm*}[t]
  \caption{Conflict-driven clause learning with parity reasoning}
  \label{alg:CDCL}

  \begin{algorithmic}[1]
    \Procedure{solve}{$\formf$}
        \State $\cdcltrail \gets \emptyset$ 
        \State \label{line:xor-detection}
        $G \gets$ \parityreasoningformat{detectParities}($\formf$)
        \State $\cdcldatabase \gets \formf$ 
        \While{True}
        \State 
        ($\pclit, \pcreason$) $\gets$ 
        propagate($\cdcldatabase$, $\cdcltrail$)
        \If{$\pclit = \pcnull$} 
        ($\pclit, \pcreason$) $\gets$ \label{line:xor-propagation}
        \parityreasoningformat{propagateXOR}($G$, $\cdcltrail$)\EndIf
        \If{$\pclit = \pcnull$} 
        ($\pclit, \pcreason$) $\gets$ 
        (nextDecision(), $\pcnull$)\EndIf
        \If{$\pclit = \pcnull$} 
        \Return $\pcsat$ \EndIf
        \State $\cdcltrail$.push($\pclit, \pcreason$)
        \If{hasConflict($\cdcltrail$)}
        \label{line:conflict-detection}
          \State $\pclearned \gets$ analyse($\cdcltrail$)
          \If{$\pclearned = \bot$}
          \State \Return $\pcunsat$
          \Else
            \State $\cdcldatabase \gets \cdcldatabase \cup \set{\pclearned}$
            \State $\cdcltrail \gets$ backjump($\cdcltrail$)
          \EndIf
        \EndIf
        \EndWhile
    \EndProcedure
  \end{algorithmic}
\end{algorithm*}

Let us first describe CDCL without parity reasoning,
i.e., without the \parityreasoningformatdescription code on 
lines~\ref{line:xor-detection}
and~\ref{line:xor-propagation}.
When run on a formula~$\formf$, the CDCL solver has a 
\emph{database}~$\cdcldatabase$ of clauses,
which is initialized to the clauses in~$\formf$.
The solver also maintains a
\emph{trail} 
consisting of an ordered list of literals assigned to true together with
reasons for these assignments. In what follows, it will be convenient
to identify the trail with the assignment setting the literals on the
trail to true. The trail is initialized to be empty.

The solver adds assigned literals to the trail, one by one,
according to the following procedure.
If some clause 
$\pcreason \in \cdcldatabase$
unit propagates an unassigned literal~$\litell$ in the sense explained in
\refsec{sec:prelims}
(which for a 
clause~$\pcreason$ 
means that all literals in the clause
except~$\litell$ are falsified by the current trail),
then the trail is extended by adding~$\litell$ 
with~$\pcreason$ 
as the
\emph{reason clause} explaining the propagation.
(If several clauses propagate at the same time, then ties will be split in
a somewhat arbitrary fashion depending on low-level details in the
algorithm implementation.) 
Otherwise, the solver uses a decision heuristic to pick some
literal to assign to true. Such a decision literal has no reason
clause. If there is no literal left to assign, then this means that
all variables have been assigned without violating any clause
in~$\formf$. In other words, a satisfying assignment has been found,
and so the solver returns that the formula is satisfiable.
Assuming instead that some literal has been added to the trail, 
this literal can  lead to some clause~$\constrd \in \cdcldatabase$ being
falsified by the  trail. This is referred to as a 
\emph{conflict} with
$\constrd$ as the \emph{conflict clause}.
When a conflict arises, a 
\emph{conflict analysis} 
algorithm is called to derive a new 
clause~$\pclearned$ 
from  the conflict clause and the reason clauses
currently on the trail. If the result of this conflict analysis is the
empty clause~$\bot$ without any literals, then contradiction has been
derived and the solver returns that the formula is unsatisfiable.
Otherwise, the 
clause~$\pclearned$ 
is \emph{learned}, i.e., added to the
clause database, 
after which the solver
\emph{backjumps} by removing some literals from the trail
until~$\pclearned$ 
is no longer falsified. The details of exactly how
backjumps are done are not relevant for our proof logging discussion,
and there are also other details of CDCL that we are ignoring in this
description, such as that the solver sometimes does 
\emph{restarts}
(which means resetting the trail to be empty)
and sometimes performs 
\emph{database reduction}
(removing learned clauses from~$\cdcldatabase$).

To add parity reasoning to CDCL, the solver is modified by first 
detecting implicit parity constraints in the CNF formula on
line~\ref{line:xor-detection} in Algorithm~\ref{alg:CDCL}.
This can be done by checking syntactically if all clauses in the
canonical clausal encoding of a parity constraint are present. For
instance, the clauses
\begin{subequations}
  \begin{equation}
    \label{eq:parity-clauses}
    \set{
      x_1 \lor x_2 \lor x_3 
      , \ 
      \olnot{x}_1 \lor \olnot{x}_2 \lor x_3 
      , \
      \olnot{x}_1 \lor x_2 \lor \olnot{x}_3 
      , \
      x_1 \lor \olnot{x}_2 \lor \olnot{x}_3 
    }
  \end{equation}
  encode the parity constraint
  \begin{equation}
    \label{eq:parity-constraint}
    x_1 \oplus x_2 \oplus x_3 = 1
    \eqperiod
  \end{equation}
\end{subequations}
Parity constraints detected in this way can then be used for
Gaussian elimination, which generates new parity constraints.
If all variables in a parity constraint except one is assigned by the
trail, then the final variable is propagated to a value on
line~\ref{line:xor-propagation}. It can also happen that a parity
constraint is violated by the current trail, and detection of this
condition is included on line~\ref{line:conflict-detection}.
In both of these cases, the solver will need a reason or conflict
clause, respectively, to justify the steps taken. Such a clause can be
computed from the parity constraint in a straightforward way. Suppose,
for example that from parity constraints
$x_1 \oplus x_2 \oplus x_3 = 1$
and
$x_2 \oplus x_3 \oplus x_4 = 1$
Gaussian elimination has derived 
$x_1 \oplus x_4 = 0$,
and  suppose also
that $x_1$ is assigned to true on the trail. Then
$x_4$ will also propagate to true, and the reason clause provided for this
will be
$
\olnot{x}_1 \lor x_2
$.

There are many variations on how this general idea can be
implemented.  For instance, parity detection can also be run later during
the search over the clause database~$\cdcldatabase$, 
as done in 
\jncitenameurlpunctuation{\cryptominisat}{CryptoMiniSat}{}{.}
Another interesting question studied 
\jnciteinorby{YM21EngineeringPBXOR} is whether it is better to propagate all
clauses first (as in our pseudocode here) or all parity constraints
first, or if the propagation on different types of constraints  
should be interleaved. However, such aspects are not relevant to
how proof logging for parity reasoning should be designed, and our
description in Algorithm~\ref{alg:CDCL} has been chosen mainly to make 
the exposition simple.

\providecommand{\proofloggingitem}[1]{\textbf{#1:}\xspace}

To provide proof logging for CDCL solvers with Gaussian elimination,
we will need the four ingredients listed below:
\begin{enumerate}
\item
  \proofloggingitem{XOR encoding}
  An efficient encoding of parity constraints as
  linear pseudo-Boolean constraints.
\item%
  \proofloggingitem{XOR reasoning}
  A method of deriving (the pseudo-Boolean encoding of) a new parity
  constraint from existing parity constraints.
\item%
  \proofloggingitem{Reason and conflict clause generation}
  The ability to prove the validity of reason and conflict clauses
  from the pseudo-Boolean encoding of parity constraints when such
  parity constraints give rise to propagations or conflicts,
  respectively.
\item%
  \proofloggingitem{Translation from CNF}
  A way of translating clausal encodings of parity constraints
  to pseudo-Boolean form
  (which is where we will need to go beyond cutting planes by using
  extension variables and \redbasedstrengthening).
\end{enumerate}
We will describe these components in detail in the rest of this
section. In \refsec{sec:example}, we will then provide a worked-out
example to illustrate how everything comes together to yield a method
for CDCL solving with parity constraints.

\subsection{%
  Linear    
  Pseudo-Boolean Encoding of Parity Constraints}
\label{sec:proof-logging-pb-encoding}

Our encoding of parity constraints in linear pseudo-Boolean form is
based on the observation
\jnciteinorby{DGP04GeneralizingBooleanSatisfiabilityI}
that for any partial assignment to the variables
$x_1, \ldots, x_k$, 
the parity constraint
$x_1 \oplus x_2 \oplus \dots \oplus x_k = b$ 
as in~\eqref{eq:xor} is satisfiable if and only if the 
\mbox{$0-1$} integer linear equality
\begin{equation}
  \label{eq:xor-pb}
  \sum_{i \in [k]} x_i
  = b + \sum_{i\in[\lfloor k/2 \rfloor]} 2 y_i
\end{equation}
is satisfiable, 
where
$y_1, \ldots, y_{\floor{k/2}}$
are fresh variables not appearing in other constraints.
Since the variables~$y_{i}$ are otherwise unconstrained, the
right-hand side can take any even (odd) value for $b = 0$ ($b = 1$) in
the range from $0$ to $k$, and these are exactly the values that we want
to allow for $\sum_{i \in [k]} x_{i}$.
Recalling that any equality on the form~\refeq{eq:equality}
can be represented with the two inequalities
\refeq{eq:equality-ineq-1} and~\refeq{eq:equality-ineq-2},
we have obtained a representation of parity constraints as linear
pseudo-Boolean inequalities.

In fact, we can generalize this  by observing that if we
let $\s{\mathcal{B}}$ denote any integer linear combination of
variables, possibly also with a constant term, then
the two inequalities
\begin{subequations}
  \begin{align}
    \label{eq:general-XOR-PB-geq} 
    \sum_{i \in [k]} x_{i}
    &\geq b + 2 \mathcal{B}
    \\
    \label{eq:general-XOR-PB-leq} 
    \sum_{i \in [k]} -x_{i}
    &\geq - b - 2 \mathcal{B}
   \end{align}
  forming the equality 
  $\sum_{i \in [k]} x_{i} = b + 2 \mathcal{B}$       %
    imply the parity constraint
  \begin{equation}
    \label{eq:general-XOR}
    x_1 \oplus x_2 \oplus \cdots \oplus x_\s{k} = b
    \eqperiod
  \end{equation}
\end{subequations}
We will make repeated use of this observation
below.

\subsection{XOR Reasoning Using Pseudo-Boolean Constraints}
\label{sec:proof-logging-xor-reasoning}

Whenever we want to combine two  XOR constraints 
to derive a new XOR constraint  as is done
during Gaussian elimination, we only need to add the pseudo-Boolean
equalities corresponding to these two XOR constraints. Consider again
our example derivation
\begin{equation}
\AxiomC{$x_1 \oplus x_2 \oplus x_3 = 1$}
\AxiomC{$x_2 \oplus x_3 \oplus x_4 = 1$}
\BinaryInfC{$x_1 \oplus x_4 = 0$}
\DisplayProof 
\end{equation}
from before, and assume that the two premises are represented
in pseudo-Boolean form as
\begin{subequations}
  \begin{align}
    \label{eq:first-xor-PB}
    x_1 + x_2 + x_3 &= 2 y_1 + 1  
    \\
    \shortintertext{and}
    \label{eq:second-xor-PB}
    x_2 + x_3 + x_4 &= 2 y_2 + 1
    \\
    \intertext{for fresh variables~$y_1$ and~$y_2$.
    Then
    adding both equalities together we obtain}
    x_1 + 2 x_2 + 2 x_3 + x_4 &= 2 y_1 + 2 y_2 + 2
  \end{align}
\end{subequations}
which implies the desired XOR constraint by
the observation we just made regarding
\eqref{eq:general-XOR-PB-geq}--\eqref{eq:general-XOR}.
(Recall that a linear combination of equalities as
in~\eqref{eq:equality} is a notational shorthand for taking pairwise
linear combinations of inequalities~\eqref{eq:equality-ineq-1}
and~\eqref{eq:equality-ineq-2}.)

\subsection{Reason and Conflict Clause Generation
  from XOR Constraints
}
\label{sec:proof-logging-clause-generation}

As explained above, CDCL solvers justify all propagation and conflict
analysis steps using clauses. If we want to use XOR constraints to
propagate forced variable assignments or derive contradiction, then we
need to provide clauses that justify such derivation steps,  together
with proof logging steps explaining why these clauses are valid.  We
next show how to derive such clauses from pseudo-Boolean encodings of
XOR constraints.

Suppose we have a parity constraint
encoded by inequalities of the
form~\refeq{eq:general-XOR-PB-geq}--\refeq{eq:general-XOR-PB-leq},
and
let~${\rho}$
be an assignment to the 
variables $x_1, \ldots, x_k$
that is inconsistent 
with these inequalities
because it falsifies the implied XOR constraint~\eqref{eq:general-XOR}. 
We want to derive 
from~\refeq{eq:general-XOR-PB-geq}--\refeq{eq:general-XOR-PB-leq}
a clause that is falsified under ${\rho}$.
Let
\begin{equation}
  \label{eq:vars-false}
  \varsfalse{\rho} = \setdescr{i \in [k]}{\rho(x_i) = 0}  
\end{equation}
be the set of indices of variables assigned to false by~${\rho}$ and
\begin{equation}
  \label{eq:vars-true}
  \varstrue{\rho} = \setdescr{j \in [k]}{\rho(x_j) = 1}  
\end{equation}
the indices of variables  assigned  to true . Using 
the literal axiom rule~\eqref{eq:cp-rule-literal}
we can derive
(the normalized form of) the trivially true constraint
\begin{equation}
  \sum_{i \in {\varsfalse{\rho}}} x_i + 
  \sum_{ j\in {\varstrue{\rho}}} -x_j 
  \geq  -\setsize {\varstrue{\rho}}
  \eqcomma
  \label{eq:from_literal_axiom}
\end{equation}
which when added to~\eqref{eq:general-XOR-PB-geq} yields
\begin{equation}
  \label{eq:two-times-false}
  \sum_{i \in {\varsfalse{\rho}}} 2 x_i \geq b -
  \setsize{\varstrue{\rho}}
  + 2{\mathcal{B}} 
  \eqperiod{}
\end{equation}
By assumption, we have that
$b - \setsize{\varstrue{\rho}}$ 
is odd, since otherwise
$\rho$ would not falsify the XOR constraint implied
by~\eqref{eq:general-XOR-PB-geq}--\eqref{eq:general-XOR-PB-leq}.
All other terms in the inequality~\refeq{eq:two-times-false} are
divisible by~$2$. 
Hence, even though~\eqref{eq:two-times-false} 
is not presented in normalized form, we can see that if we apply the
division rule~\eqref{eq:cp-rule-div} with divisor~$2$,
this will round up and increase the  degree of falsity. 
This means that if we divide the constraint~\eqref{eq:two-times-false}
and then  multiply by~$2$
(which is just a special case of the linear combination
rule~\eqref{eq:cp-rule-lincomb}), we get
\begin{equation}
    \sum_{i \in {\varsfalse{\rho}}} 2 x_i \geq b -
    \setsize{\varstrue{\rho}} + 1 + 2 {\mathcal{B}} \eqperiod{}
\end{equation}
We continue by adding~\eqref{eq:general-XOR-PB-leq} to get
\begin{equation}
   \sum_{i \in {\varsfalse{\rho}}} x_i - \sum_{j \in {\varstrue{\rho}}} x_j \geq
   1 - \setsize{\varstrue{\rho}}
   \eqcomma
\end{equation}
which is the same constraint as
\begin{equation}
  \label{eq:final-reason-clause}
  \sum_{i \in {\varsfalse{\rho}}} x_i + \sum_{j \in {\varstrue{\rho}}} \olnot{x}_j \geq   1
\end{equation}
after normalization.
This last constraint, which is a disjunctive clause, is falsified under $\rho$
as desired, and so can serve as the conflict clause justifying 
why the assignment~${\rho}$ is inconsistent.
Derivations of reason clauses for propagation work in a similar
way---essentially, we can pretend that the propagated variable is set
to the wrong value and then perform the derivation above to obtain a
clause~\eqref{eq:final-reason-clause} that propagates the variable to
the right value instead. 
An example for deriving a reason clause can be found 
towards
the end of \refsec{sec:example}.

\subsection{Translating Parity Constraints from CNF
  to Pseudo-Boolean Form}
\label{sec:proof-logging-translation}

An XOR constraint as in~\eqref{eq:xor} can be encoded into CNF 
in a canonical way by 
including for each of the $2^{{k}-1}$ assignments falsifying the 
constraint the disjunctive clause ruling out that assignment.
For example, for $k = 3$ and $b = 1$
the parity constrain~\eqref{eq:parity-constraint}
can be encoded by the clauses in~\eqref{eq:parity-clauses},
which are written as the \mbox{$0$-$1$} integer linear inequalities
\begin{subequations}
  \begin{align}
    {x}_1 + {x}_2 + {x}_3 &\geq 1 \\
    \olnot{{x}}_1 + \olnot{{x}}_2 + {x}_3 &\geq 1 \\
    \olnot{{x}}_1 + {x}_2 + \olnot{{x}}_3 &\geq 1 \\
    {x}_1 + \olnot{{x}}_2 + \olnot{{x}}_3 &\geq 1
  \end{align}
\end{subequations}
in pseudo-Boolean form.
Since the number of clauses in this canonical CNF encoding of an XOR
constraint scales exponentially with the number of variables,
it is only feasible to encode short XORs into CNF in this manner.
However, it is possible to split up a
long XOR constraint into multiple constant-size XORs using auxiliary
variables $z_{i}$, which represent the partial 
parities up to and
including ${x}_{i}$, i.e., 
$z_i = \sum_{{j}\in[i]} {x}_{j}\ (\bmod\ 2)$.
In this way, a collection of parity constraints
\begin{subequations}
\begin{align}
  \label{eq:long-xor-1}
  {x}_1 \oplus {x}_2 \oplus {z}_2 &= 0 \\
  \label{eq:long-xor-2}
  {z}_2 \oplus {x}_3 \oplus {z}_3 &= 0 \\
  \nonumber
  &\hspace{0.60em}\vdots \\
  \label{eq:long-xor-last}
  {z}_{k-2} \oplus {x}_{{k}-1} \oplus {x}_{k} &= b 
\end{align}
\end{subequations}
can be used to represent the constraint~\eqref{eq:xor}.
Assuming that we can split up parity constraints in this manner,
we will only need to translate short parity constraints from CNF to
pseudo-Boolean form. The original, long, parity
constraints can then be recovered by XOR reasoning, just summing up the
\mbox{constraints~\refeq{eq:long-xor-1}--\refeq{eq:long-xor-last}},
and proof logging for this derivation can be done as described in
\refsec{sec:proof-logging-xor-reasoning}
above.

We perform the translation to the pseudo-Boolean XOR encoding from CNF
in two steps, which we will describe in more detail after providing
the general idea.
The first step is to derive the constraint
\begin{equation}
  \label{eq:xor-recover}
  \sum_{{i} \in [{k}]} {x}_{i} 
  = 
  \sum_{{i}\in[\lfloor{k}/2\rfloor]} 2 {y}_{i} + y' 
  \eqcomma 
\end{equation}
where~$y'$ and 
$y_{i}$, $i \in [\lfloor{k}/2\rfloor]$, 
are all
fresh variables. Note that adding the equality
constraint~\eqref{eq:xor-recover} to any formula does not affect
satisfiability, because we can always assign the fresh variables so
that this additional constraint holds true.
However, although the constraint~\refeq{eq:xor-recover} is redundant
in the sense of Proposition~\ref{prop:redundant}, we cannot use the
\redbasedstrengthening rule to derive the constraint, because we do
not have an efficient procedure for constructing a
witness~$\mapassmnt$ that is efficiently verifiable by
Algorithm~\ref{alg:PBRAT}.
Instead, we will introduce the auxiliary variables
$y'$ and $y_{i}$, $i \in [\lfloor{k}/2\rfloor]$,
one by one, in a similar fashion to what
was done in \refsec{sec:mapping-redundancy}.
We remark that an alternative to~\eqref{eq:xor-recover} would be to
encode the sum 
$
2 \cdot 
\Floor{
  \frac{1}{2}
  \sum_{i \in [k]} x_i 
}
$
of the $x_i$-variables rounded down to the nearest even integer
as a sum of powers of~$2$, resulting in an equality constraint
\begin{equation} 
  \sum_{i \in [k]} x_{i} 
  = 
  \sum_{i\in[\lceil\log_2(k/2)\rceil]} 2^{i} {y}_{i} + y' 
  \eqperiod
\end{equation}
For parity constraints over a large number of variables, 
this encoding
has a substantially smaller number of auxiliary variables. However, since
we are recovering parity constraints from CNF, we only expect to have
parity constraints over few variables, as the number of
clauses in the CNF encoding is exponential in the number of variables.

The second step, once we have derived the equality
constraint~\eqref{eq:xor-recover}, is to brute-force over all possible
assignments to the  $x_{i}$-variables to derive the equality
\begin{equation}
  \label{eq:y-prime-equals-b}
  y' = b
  \eqperiod
\end{equation}
Summing the equalities~\eqref{eq:xor-recover} and~\refeq{eq:y-prime-equals-b},
we obtain a constraint of the desired
form~\refeq{eq:xor-pb}.
Note that since we are considering all possible assignments to the 
$x_{i}$-variables, this derivation will require an exponential number of 
derivation steps
measured in the number of variables,
but this is still polynomial measured in the number of clauses in the canonical
CNF encoding of parity constraints. 
We now proceed to describe  this process in detail.

\begin{figure}[t]
  \begin{subfigure}[b]{0.40\textwidth}
    \centering
    \includegraphics[scale=0.4]{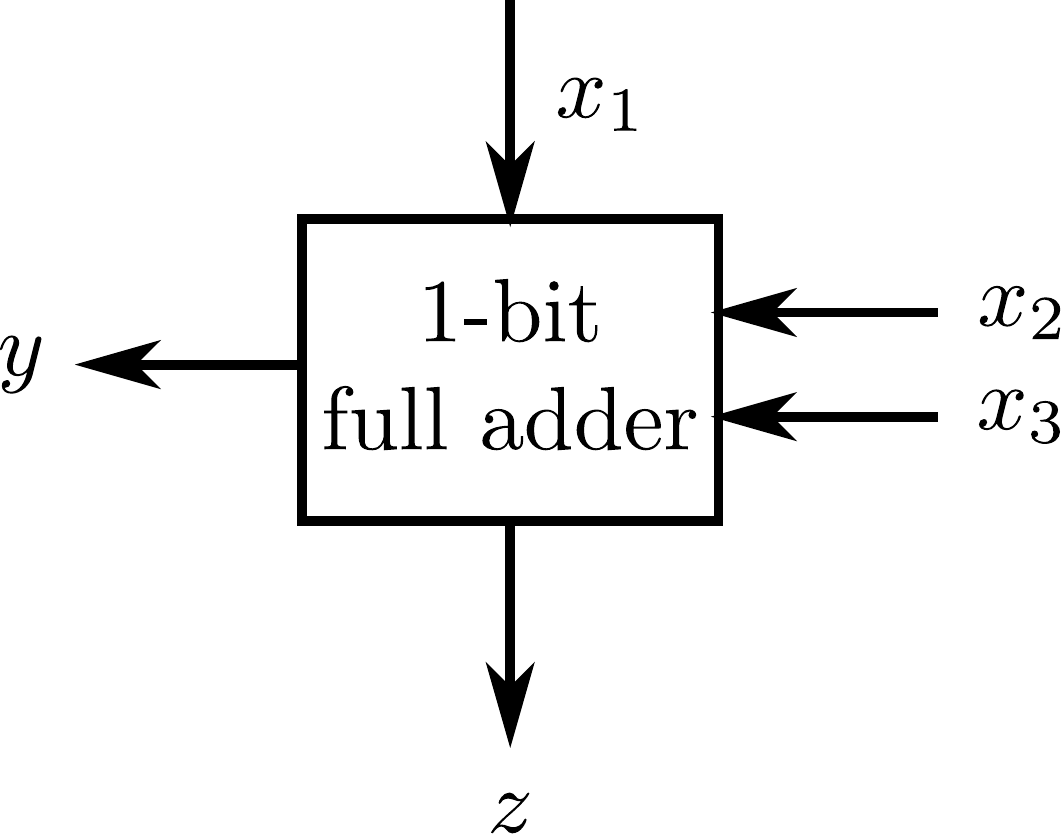}
    \caption{A 1-bit full adder.}
    \label{fig:full-adder}
  \end{subfigure}
  \hfill
  \begin{subfigure}[b]{0.40\textwidth}
    \centering
    \includegraphics[scale=0.4]{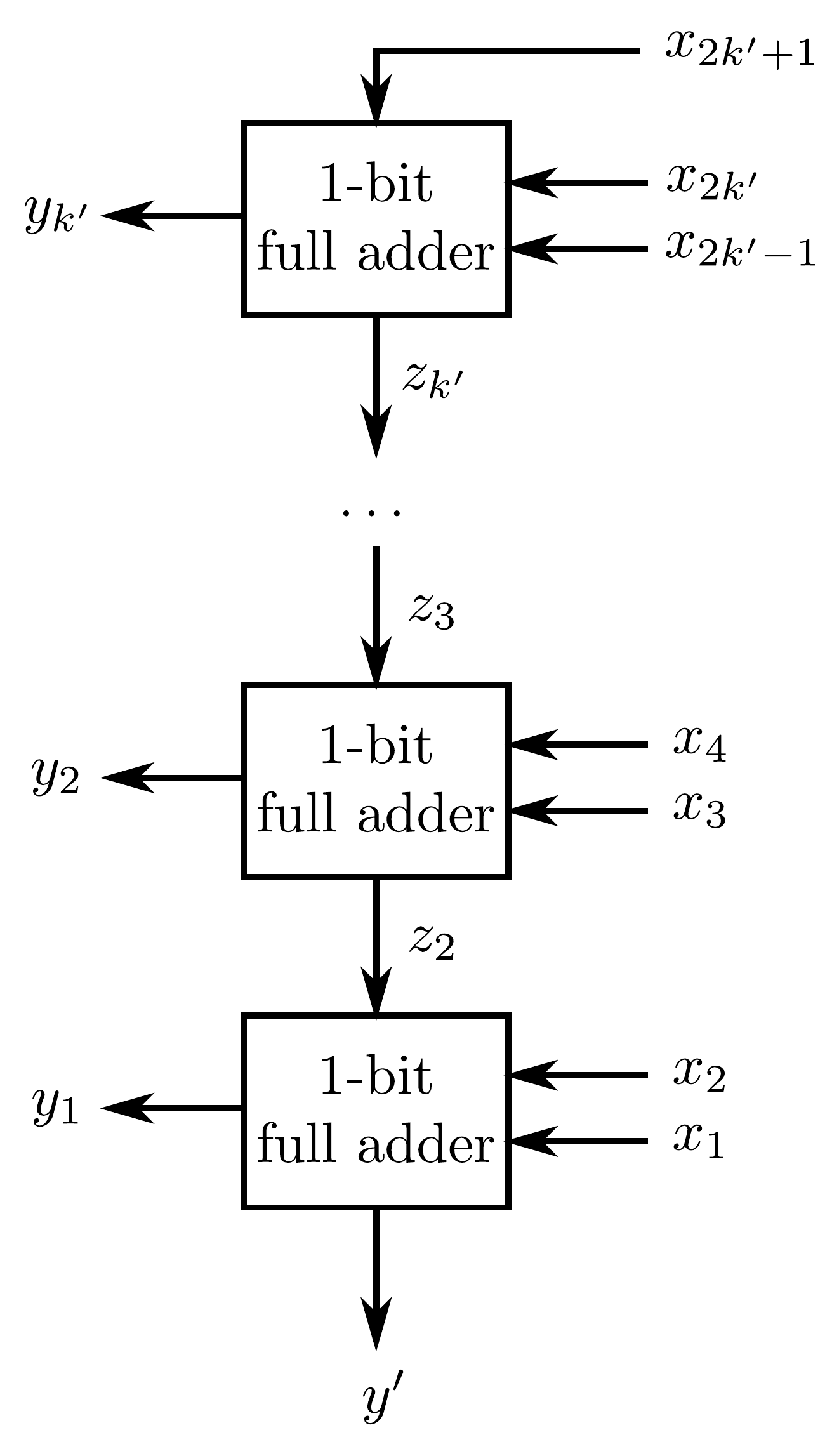}
    \caption{Chain of 1-bit full adders.}
    \label{fig:full-adder-chain}
  \end{subfigure}
  \caption{%
    Using adders to encode parities of subsets of variables.
  }
\end{figure}

\begin{description}
\item [Step 1a:] 
  To derive
  \eqref{eq:xor-recover}
  we will construct a chain of $1$-bit full adders, as illustrated in
  \reffig{fig:full-adder-chain}
  for an adder with output carry bit~$y$ and sum bit~$z$. 
  Let us start  by showing how
  the encoding of a single adder can be derived. 
  A $1$-bit full adder
  (shown in \reffig{fig:full-adder}) computes the sum of three variables
  $x_1, x_2, x_3$ and returns the result as a binary number. This can be
  encoded using the pseudo-Boolean equality
  \begin{equation}
    \label{eq:full-adder}
    2y + z = x_1 + x_2 + x_3
    \eqperiod 
  \end{equation}
  Recalling the shorthand~\refeq{eq:reification-notation}
  for the two reification constraints~\refeq{eq:reification-encoding-1}
  and~\refeq{eq:reification-encoding-2}, in order to
  obtain~\eqref{eq:full-adder} we start by deriving 
\begin{subequations}
\begin{align}
  y &\reifequiv x_1 + x_2 + x_3 \geq 2 
  \\
  z &\reifequiv x_1 + x_2 + x_3 - 2y \geq 1
\end{align}
\end{subequations}
for fresh variables~$y$ and~$z$
using \redbasedstrengthening as described in
Proposition~\ref{prop:reification-derivation}. 
This means that we have now derived the four constraints
\begin{subequations}
\begin{align}
  \label{eq:def-y-if}
  2\olnot{{y}} +
  {x}_1 + {x}_2 + {x}_3 
  &\geq 2 
  \\
  \label{eq:def-y-only-if} 
  2 y+
  \olnot{{x}}_1 + \olnot{{x}}_2 + \olnot{{x}}_3 
  &\geq 2 
  \\
  \label{eq:def-z-if}
  3 \olnot{{z}}  + 
  {x}_1 + {x}_2 + {x}_3 + 2\olnot{{y}} 
  &\geq 3 
  \\
  \label{eq:def-z-only-if}
  3 {z} +
  \olnot{{x}}_1 + \olnot{{x}}_2 + \olnot{{x}}_3 + 2{y}
  &\geq 3                                                                
\end{align}
\end{subequations}
when written as pseudo-Boolean inequalities in normalized form.
To derive
the less-than-or-equal part
$
2y + z \leq x_1 + x_2 + x_3
$
of~\eqref{eq:full-adder},
which in normalized form is
\begin{subequations}
\begin{equation}
  \label{eq:full-adder-pb-leq}
  {x}_1 + {x}_2 + {x}_3 + 2\olnot{{y}} + \olnot{{z}} \geq 3
  \eqcomma
\end{equation}
we take a linear combination of~\eqref{eq:def-z-if} and $2$
times~\eqref{eq:def-y-if}, followed by division by~$3$.  
In a similar fashion,
to derive the greater-than-or-equal part
$
2y + z \geq x_1 + x_2 + x_3
$
of~\eqref{eq:full-adder}, or
\begin{equation}
  \label{eq:full-adder-pb-geq}
  \olnot{{x}}_1 + \olnot{{x}}_2 + \olnot{{x}}_3 + 2{y} + {z}
  \geq 3
\end{equation}
\end{subequations}
in normalized form,
we add together~\eqref{eq:def-z-only-if} and $2$
times~\eqref{eq:def-y-only-if} followed by division by~$3$. 

\newcommand{\kprime}{k'}

   \item[Step 1b:] 
     To derive the equality constraint~\eqref{eq:xor-recover}, we use a
     chain of $1$-bit full adders connected as in
     \reffig{fig:full-adder-chain}, where we set
     $k' = \floor{k/2}$.
     The
     $x_i$-variables are used as inputs to the adders, and the final
     variable~$x_{2 \kprime + 1}$, which appears in the
     topmost adder, will only be there if the number of variables~$k$
     is odd.
     Otherwise, we replace~$x_{2 \kprime + 1}$ by~$0$,
     so that the topmost adder only has $x_{2 \kprime}$ 
     and~$x_{2 \kprime - 1}$ as input. 
     (Formally, if $x_{2 \kprime + 1}$ does not exist, then it is a
     fresh variable,
     and so we can derive the equality~$x_{2 \kprime + 1} = 0$ by
     \redbasedstrengthening 
     before continuing as described below.)
     The output carry variables~$y_i$, $i \in [k']$,
     and the final sum bit~$y'$
     will be used to derive the equality~\eqref{eq:xor-recover}, while
     the $z_i$-variables
     are intermediate parity bits. 
     We apply the procedure in Step~1a to all adders to derive PB
     constraints on the form~\refeq{eq:full-adder-pb-leq}
     and~\refeq{eq:full-adder-pb-geq}.
     After this, for the topmost adder in
     \reffig{fig:full-adder-chain} we have obtained
     \begin{subequations}
       \begin{equation}
         \label{eq:top-adder}
         2y_{\kprime} + z_{\kprime} = x_{2{\kprime}+1} + x_{2{\kprime}}
         + x_{2{\kprime}-1}
         \eqcomma
       \end{equation}
       for the intermediate adders the equations
       \begin{equation}
         2y_i + z_i = z_{i+1} + x_{2i} + x_{2i - 1} \label{eq:middle-adder}
       \end{equation}
       hold for $i\in \set{2,\dots,{\kprime}-1}$,
       and for the bottom adder we get
       \begin{equation}
         2y_1 + y' = z_{2} + x_{2} + x_{1} \label{eq:middle-bottom}
         \eqperiod{}
       \end{equation}
     \end{subequations}
   By adding the encoding of all $1$-bit adders, i.e.,
   the equalities~\eqref{eq:top-adder}--\eqref{eq:middle-bottom}, we obtain
   \begin{equation}
     \label{eq:almost-xor}
    \sum_{i=1}^{\kprime} 2 y_i + y' + \sum_{i=2}^{\kprime} z_i 
    = 
    \sum_{i=2}^{\kprime} z_i + \sum_{i=1}^{2{\kprime} + 1} x_i
    \eqcomma
\end{equation}
where the sums $\sum_{i=2}^{\kprime} z_i$ on each side cancel to produce
the equality constraint~\eqref{eq:xor-recover}  as desired.

\item[Step 2:] 
  The final step is to fix the value of ${y'}$ in order to go
  from~\eqref{eq:xor-recover} to our final goal~\refeq{eq:xor-pb}.
  That is, writing
  $\yprimesigned = y'$
  if $b=1$
  and
  $\yprimesigned = \olnot{y}'$ if $b=0$,
  we wish to derive the PB constraint
  \begin{equation}
    \label{eq:yprime-unit-constraint}
    \yprimesigned \geq 1
  \end{equation}
  forcing $y'  = b$.
  We will do so by considering all possible truth value assignments~$\rho$
  to the variables 
  $x_i$, $i \in [k]$. In order to present the formal
  derivation, we first need to set up some notation.

  For any assignment~$\rho$ to a subset of the variables
  $x_i$, $i \in [k]$, 
  let 
  $\varsfalse{\rho}$
  and
  $\varstrue{\rho}$
  be the indices of variables set to false and true by~$\rho$,
  respectively, as defined in~\refeq{eq:vars-false} and~\refeq{eq:vars-true}.
  Let us write $\clnotass{\rho}$ to denote the unique clausal constraint
  \begin{equation}
    \label{eq:cl-not-assmt}
    \sum_{i \in \varsfalse{\rho}} x_i + 
    \sum_{j \in \varstrue{\rho}} \olnot{x}_j 
    \geq 1
  \end{equation}
  over all variables assigned  by~$\rho$ that is falsified by this assignment.
  We also extend this notation in the natural way to let
  $\clnotassrhoy$
  denote the clausal constraint
  \begin{equation}
    \label{eq:cl-not-assmt-incl-y}
    \yprimesigned + 
    \sum_{i \in \varsfalse{\rho}} x_i + 
    \sum_{j \in \varstrue{\rho}} \olnot{x}_j 
    \geq 1
  \end{equation}
  that is falsified by~$\rho$ if in addition $\yprimesigned$ is
  set to false,
  i.e., 
  $y'$ is given the value $1 - b$.

  For any assignment~$\rho$ such that
  $\sum_{i \in [k]} x_{i} \neq b \ (\bmod\ 2)$,
  we postulated above that
  the clause $\clnotass{\rho}$ is in the formula,
  but for our argument here we only need the slightly weaker
  assumption that this clause can be obtained by reverse unit
  propagation on the constraints derived so far.
  Assuming that this holds, we can certainly derive~$\clnotassrhoy$
  by RUP for all such assignments~$\rho$.
  If instead $\rho$ is such that
  $\sum_{i \in [k]} x_{i} = b \ (\bmod\ 2)$,
  then extending~$\rho$ by setting
  $\yprimesigned = 0$
  means that~\eqref{eq:xor-recover} can no longer be satisfied, since
  $\rho$~assigns different parities to the left-hand and right-hand
  sides of this equality, and no assignment to the $y_i$-variables in
  $
  \sum_{{i}\in[\lfloor{k}/2\rfloor]} 2 y_{i}
  $
  can change this.
  For such~$\rho$ we can therefore proceed as in
  \refsec{sec:proof-logging-clause-generation}
  to derive the clause~$\clnotassrhoy$
  explaining why the assignment
  $\rho \,\cup\, \set{\yprimesigned \mapsto 0}$
  is inconsistent.

  So far, we have shown how to derive clauses~$\clnotassrhoy$
  in~\refeq{eq:cl-not-assmt-incl-y}  
  for any assignment~$\rho$ to all the  $x_i$-variables.
  But once we have these clauses, the rest is routine.  Let
  $\rho_{k-1}$ be any partial
  assignment to the $k-1$ first variables  $x_i$, $i \in [k-1]$.
  Taking the previously derived constraints
  $\clnotassextxy{\rho_{k-1}}{x_k}{0}$
  and~$\clnotassextxy{\rho_{k-1}}{x_k}{1}$,
  which is what we write by mild abuse of notation to denote the
  clausal   constraints
  \begin{subequations}
    \begin{align}
    \label{eq:cl-not-assmt-incl-y-xk-false}
    \yprimesigned + x_k + 
    \sum_{i \in \varsfalse{\rho_{k-1}}} x_i + 
    \sum_{j \in \varstrue{\rho_{k-1}}} \olnot{x}_j 
    &\geq 1
    \\
    \shortintertext{and}
    \label{eq:cl-not-assmt-incl-y-xk-true}
    \yprimesigned + \olnot{x}_k + 
    \sum_{i \in \varsfalse{\rho_{k-1}}} x_i + 
    \sum_{j \in \varstrue{\rho_{k-1}}} \olnot{x}_j 
    &\geq 1
      \eqcomma
  \end{align}
\end{subequations}
  respectively
  (which agree on all literals except that the variable~$x_k$ appears
  with opposite signs),
  adding these constraints, and then dividing by~$2$ yields
  \begin{equation}
    \label{eq:cl-not-assmt-xk-eliminated}
    \yprimesigned +
    \sum_{i \in \varsfalse{\rho_{k-1}}} x_i + 
    \sum_{j \in \varstrue{\rho_{k-1}}} \olnot{x}_j 
    \geq 1
    \eqcomma
\end{equation}
i.e., the clause~$\clnotassrhoy[\rho_{k-1}]$.%
\footnote{For readers knowledgeable in proof complexity, what we are
  doing here is just the cutting planes simulation of 
  a resolution step   resolving the
  two clauses~\refeq{eq:cl-not-assmt-incl-y-xk-false}
  and~\refeq{eq:cl-not-assmt-incl-y-xk-true}
  over~$x_k$
  to obtain the clause~\refeq{eq:cl-not-assmt-xk-eliminated}.
  And, jumping ahead a bit, the whole derivation presented here is an
  adaptation of the  standard resolution derivation of contradiction from the 
  $2^k$ clauses~$\clnotass{\rho}$
  for all assignments~$\rho$ to a set of $k$~variables.}
We can eliminate the variable~$x_k$ in this way
by deriving clauses~\refeq{eq:cl-not-assmt-xk-eliminated}
for all assignments  $\rho_{k-1}$ to  $x_i$, $i \in [k-1]$.
(A technical side note is that
the constraint~\refeq{eq:cl-not-assmt-xk-eliminated} follows by
reverse unit propagation  
on~\refeq{eq:cl-not-assmt-incl-y-xk-false}
and~\refeq{eq:cl-not-assmt-incl-y-xk-true},
and so we could avoid a syntactic derivation by just claiming it as a
RUP constraint.
However, when there is a simple explicit derivation like above it is
often preferable to use such a derivation instead, since this tends to make
proof verification faster, and as we will see in \refsec{sec:example}
there is an elegant way of chaining all derivations of this type
together on a single proof line.)

Next, we consider all assignments  $\rho_{k-2}$ to  $x_i$, $i \in [k-2]$,
and repeat the derivation of clauses~\refeq{eq:cl-not-assmt-xk-eliminated}
from~\refeq{eq:cl-not-assmt-incl-y-xk-false}
and~\refeq{eq:cl-not-assmt-incl-y-xk-true}
to obtain clauses~$\clnotassrhoy[\rho_{k-2}]$
for all~$\rho_{k-2}$
(where we replace $x_k$ by~$x_{k-1}$
in~\refeq{eq:cl-not-assmt-incl-y-xk-false}
and~\refeq{eq:cl-not-assmt-incl-y-xk-true}).
Continuing in this fashion, we eliminate the variables
$
x_{k},
x_{k-1},
\ldots ,
x_{1}
$
one by one,
until the process terminates with the desired
constraint~\refeq{eq:yprime-unit-constraint}.
Since we also know
$\yprimesigned \leq 1$
(which is a literal axiom),
we now have the equality
$y' = b$ in~\refeq{eq:y-prime-equals-b},
so that we can add together~\eqref{eq:xor-recover}
and~\refeq{eq:y-prime-equals-b}.
\end{description}

This concludes our derivation of the pseudo-Boolean
encoding~\eqref{eq:xor-pb} of the 
XOR constraint~\eqref{eq:xor}
from a CNF encoding.

\newcommand{\cidno}[1]{\text{\textsc{(id: #1)}}}

\section{A Worked-Out Proof Logging Example}
\label{sec:example}

In this section, we present a concrete (toy) application
of the methods developed in \refsec{sec:proof-logging},
using this example to also illustrate the syntax used in 
\jncitenameurl{\veripb}{VeriPB}{}
proof logging files.

Suppose that we have a CNF formula with two 
parity constraints 
$\varx_1 \oplus \varx_2 \oplus \varx_3 = 0$
and
$\varx_2 \oplus \varx_3 \oplus \varx_4 = 1$, 
which are encoded as sets of clauses
\begin{subequations}
  \begin{equation}
    \label{eq:parity-clauses-123}
    \set{
      \olnot{x}_1 \lor x_2 \lor x_3 
      , \ 
      x_1 \lor \olnot{x}_2 \lor x_3 
      , \ 
      x_1 \lor x_2 \lor \olnot{x}_3 
      , \ 
      \olnot{x}_1 \lor \olnot{x}_2 \lor \olnot{x}_3 
    }
  \end{equation}
  and
  \begin{equation}
    \label{eq:parity-clauses-234}
    \set{
      x_2 \lor x_3 \lor x_4 
      , \
      x_2 \lor \olnot{x}_3 \lor \olnot{x}_4 
      , \
      \olnot{x}_2 \lor x_3 \lor \olnot{x}_4 
      , \ 
      \olnot{x}_2 \lor \olnot{x}_3 \lor x_4 
    }
    \eqcomma
  \end{equation}
\end{subequations}
respectively.
To present this formula to \veripb, we write the constraints in
pseudo-Boolean form in an input file as  
\begin{Verbatim}[commandchars=\\\{\}]
* #variable= 4 #constraint= 8
+1 \verbtilde{x1} +1  x2 +1  x3 >= 1 ;
+1  x1 +1 \verbtilde{x2} +1  x3 >= 1 ;
+1  x1 +1  x2 +1 \verbtilde{x3} >= 1 ;
+1 \verbtilde{x1} +1 \verbtilde{x2} +1 \verbtilde{x3} >= 1 ;
+1  x2 +1  x3 +1  x4 >= 1 ;
+1  x2 +1 \verbtilde{x3} +1 \verbtilde{x4} >= 1 ;
+1 \verbtilde{x2} +1  x3 +1 \verbtilde{x4} >= 1 ;
+1 \verbtilde{x2} +1 \verbtilde{x3} +1  x4 >= 1 ;
\end{Verbatim}
using the standard OPB file
format~\cite{RM16OPB}.%
\footnote{In fact, \veripb uses a slight extension of the OPB format,
  which among other things provides greater flexibility in choosing
  variable names,  but since this is not really relevant for this
  discussion we ignore   such details.}
In the \emph{proof log} file presented to the \veripb verifier, the start
of the file
\begin{Verbatim}[commandchars=\\\{\}]
pseudo-Boolean proof version 1.1
f 8
\end{Verbatim}
instructs the verifier to read this input file, and to expect to see
$8$~constraints. The verifier maintains a database of pseudo-Boolean
constraints, and keeps track of the constraints in the database by
assigning to each constraint a  \emph{constraint identifier}, which is
a positive integer. 
Upon reading the input file above, the verifier
will assign the pseudo-Boolean constraints in the file identifiers~$1$
through~$8$ and store them in the constraints database.

When the SAT solver execution starts, the solver
reads the formula consisting of all of these clauses from file,%
\footnote{Although the SAT solver would instead expect its input to be
  formatted according to the standard 
  \jncitenameurl{DIMACS format}{DIMACSformatCNF}{}
  used in the 
  \jncitenameurlpunctuation
  {SAT competitions}
  {SATcompetition}
  {}
  {.}
  Translating a
  CNF formula from DIMACS format to OPB format is a simple syntactic
  operation, and we ignore this detail here.
}
and then runs an algorithm to detect clausal encodings of parities.
Once the SAT solver detects a parity constraint, it generates a derivation
of the pseudo-Boolean encoding of this constraint and writes it to
the proof log file.
For this translation from CNF to pseudo-Boolean form  it is necessary
to introduce fresh variables using \redbasedstrengthening.
For each application of the \redbasedstrengthening rule, 
the proof log will contain a line of the form
\begin{Verbatim}[commandchars=\\\{\}]
red [constraint C] ; [assignment omega]
\end{Verbatim}
where 
\verb+red+
identifies the line as a
\redbasedstrengthening step,
followed by the constraint $\constrc{}$ to be added and the witness
substitution~$\mapassmnt{}$. 
The substitution~$\mapassmnt$ is specified by listing each variable in
the domain of~$\mapassmnt$ followed by the value or literal it should
be substituted by, optionally separated by 
``\verb+->+''.

The translation of the clausal encoding of the parity 
$\varx_1 \oplus \varx_2 \oplus \varx_3 = 0$
into  linear pseudo-Boolean form  starts with the 
reification $\vary_1 \reifequiv \varx_1
+ \varx_2 + \varx_3 \geq 2$ for the fresh variable~$\vary_1$. 
The two PB constraints in the reification 
(which are of the form~\refeq{eq:reification-example-encoding-1}
and~\refeq{eq:reification-example-encoding-2})
are derived by the two lines
\begin{Verbatim}[commandchars=\\\{\}]
red +2 \verbtilde{y1} +1  x1 +1  x2 +1  x3 >= 2 ; y1 -> 0
red +2  y1 +1 \verbtilde{x1} +1 \verbtilde{x2} +1 \verbtilde{x3} >= 2 ; y1 -> 1
\end{Verbatim}
in the proof file, 
which can be checked using Algorithm~\ref{alg:PBRAT} as 
shown in Proposition~\ref{prop:reification-derivation}.
After the verifier has succeeded in validating these
\redbasedstrengthening steps, it adds the new constraints
\begin{subequations}
\begin{align}
    \cidno{9}
  && 
     2 \olnot{y}_1 +  x_1 +   x_2 +   x_3 
  &\geq 2
  \\
  \cidno{10}
  && 
     2  y_1 + \olnot{x}_1 + \olnot{x}_2 + \olnot{x}_3
  &\geq 2
\end{align}
\end{subequations}
to the constraint database (where we note that the constraints are
assigned identifiers~$9$ and~$10$). 
In a completely analogous fashion, 
the reification
$\vary_2 \reifequiv \varx_1 + \varx_2 + \varx_3 - 2\vary_1 \geq 1$ 
can be derived by the lines
\begin{Verbatim}[commandchars=\\\{\}]
red +3 \verbtilde{y2} +1  x1 +1  x2 +1  x3 +2 \verbtilde{y1} >= 3 ; y2 -> 0
red +3  y2 +1 \verbtilde{x1} +1 \verbtilde{x2} +1 \verbtilde{x3} +2  y1 >= 3 ; y2 -> 1
\end{Verbatim}
in the proof log, 
which adds the constraints
\begin{subequations}
\begin{align}
    \cidno{11}
  && 
     3 \olnot{y}_2  +  
     x_1 +   x_2 +   x_3   +   2 \olnot{y}_1 
  &\geq 3
  \\
  \cidno{12}
  && 
     3 y_2  +  
     \olnot{x}_1 +   \olnot{x}_2 +   \olnot{x}_3   +  2 y_1 
  &\geq 3
\end{align}
\end{subequations}
to the database of the verifier.

Once these proof logging steps have been performed, 
the variables $\vary_1$ and~$\vary_2$ correspond to the
output carry and sum bit, respectively, of
a single full adder with inputs $x_1, x_2, x_3$.
Since the parity is only over three variables in this example,
we do not need to derive a chain of multiple full adders as explained in
\refsec{sec:proof-logging-translation}.
It is important to note that the SAT solver database will not contain
any of these constraints---indeed, the CDCL algorithm does not even know
what a ``pseudo-Boolean constraint'' is---but that the PB
constraints only exist in the verifier constraint database for proof
logging purposes. However, it is important that the verifier maintains
a separate name space for auxiliary variables like $y_1$ and~$y_2$, so
that the SAT solver will not try to use the same variables for any
preprocessing or inprocessing steps. If this happens, this will most
likely result in an incorrect proof, making the verifier reject.

\DefineShortVerb{|}
The next step 
in our proof logging example
is to combine the constraints we just derived with
identifiers~$9$ through~$12$ via a sequence of cutting planes rule
applications. The version of cutting planes in \veripb is slightly
different from  
(but equivalent to) what we described in \refsec{sec:prelims} in that
the derivation rules are addition, scalar multiplication, and
division.%
\footnote{And there is also an additional saturation rule, just as
  described in \refsec{sec:prelims}, but we will not read the
  saturation rule for  this proof logging example.}
Such operations are written in postfix notation 
(also known as reverse polish notation)
in the following way:
\begin{itemize}
\item 
  To use a literal axiom
  $y \geq 0$
  or
  $\olnot{y} \geq 0$,
  we simply write
  ``\Verb|y|''
  or
  ``\Verb[commandchars=\\\{\}]|\verbtilde{y}|'',
  respectively.

\item 
  To add two constraints with identifiers
  $\mathit{id1}$
  and
  $\mathit{id2}$, 
  we write
  ``\verb|id1 id2 +|''.
\item 
  To multiply a constraint with identifier
  $\mathit{id}$
  by a positive integer $c$,   we write
  ``\verb|id c *|''.
\item 
  To divide  a constraint 
  $\mathit{id}$
  by a positive integer $c$,   we write
  ``\verb|id c d|''.
\end{itemize}
Arbitrary combinations of such derivation steps can be
performed using the 
reverse polish notation rule
in \veripb,
written on a line in the proof log prefixed by \verb+p+
(or~\verb+pol+),
where the semantics is that
any operands (constraint identifiers or factors/divisors) are
pushed on a stack, and operators pop the top two elements from this
stack and then push back the result of the operation. The final constraint
resulting from a sequence of operations is stored with the next
available constraint identifier number. In our example, the next lines
in the proof log will be
\begin{Verbatim}[commandchars=\\\{\}]
p 11  9 2 * + 3 d
p 12 10 2 * + 3 d
\end{Verbatim}
where
the first line starts with constraint number~$11$ and adds $2$~times
the constraint~$9$, after which the result is divided by $3$ (and
rounded up). The same operations are done in the second line but with the
constraints with identifiers $12$ and~$10$. The two lines derive the constraints
\begin{subequations}
\begin{align}
    \cidno{13} && \varx_1 + \varx_2 + \varx_3 + 2 \olnot{\vary}_1 + \olnot{\vary}_2 &\geq 3
    \label{eq:example_pre_parity_1}\\
    \cidno{14} && \olnot{\varx{}}_1 + \olnot{\varx{}}_2 + \olnot{\varx{}}_3 + 2 \vary_1 + \vary_2 &\geq 3
    \label{eq:example_pre_parity_2}
\end{align}
\end{subequations}
encoding an equality of the form~\eqref{eq:xor-recover}.

The inequalities in 
\eqref{eq:example_pre_parity_1}--\eqref{eq:example_pre_parity_2} 
do not yet enforce any parity
constraint on the variables $x_1, x_2, x_3$, 
since the fresh variables
$\vary_1$ and~$\vary_2$ 
are unconstrained and can be made to satisfy the constraints for any
values assigned to $x_1, x_2, x_3$.
To address this, we need to fix the value 
of~$\vary_2$, which can be done by generating a brute-force
derivation of
$\olnot{\vary{}}_2 \geq 1$
as described in \refsec{sec:proof-logging-translation}.
Since in our example we are dealing only with \mbox{$3$-XORs}, 
i.e., parity constraints over only three variables, we can take
a little shortcut when generating the missing clausal
constraints of the form~\refeq{eq:cl-not-assmt-incl-y}.
If we assign the variables $x_1, x_2, x_3$ so that the parity is even
but set $y_2=1$, then the 
constraints~\eqref{eq:example_pre_parity_1}--\eqref{eq:example_pre_parity_2} 
will propagate to contradiction since there is only a single
variable~$y_1$ left. This means that we can use
reverse unit propagation steps
\begin{Verbatim}[commandchars=\\\{\}]
rup +1 \verbtilde{y2} +1  x1 +1  x2 +1  x3 >= 1 ;
rup +1 \verbtilde{y2} +1  x1 +1 \verbtilde{x2} +1 \verbtilde{x3} >= 1 ;
rup +1 \verbtilde{y2} +1 \verbtilde{x1} +1  x2 +1 \verbtilde{x3} >= 1 ;
rup +1 \verbtilde{y2} +1 \verbtilde{x1} +1 \verbtilde{x2} +1  x3 >= 1 ;
\end{Verbatim}
to derive the clausal constraints that we need
(where each \verb+rup+-line claims that adding the negation
of the specified constraint
as in~\refeq{eq:negation}
to the current database will cause unit propagation to contradiction,
which is checked by the verifier before the constraint is added to the
database), 
and we list below these new constraints~$15$--$18$ together with the
relevant input constraints 
\begin{subequations}
\begin{align}
  \label{eq:brute-force-clause-1}
  \cidno{15}
  &&
     \olnot{y}_2 +  x_1 +  x_2 +  x_3 &\geq 1 
  \\
  \label{eq:brute-force-clause-2}
  \cidno{3}
  &&
     x_1 +  x_2 + \olnot{x}_3  &\geq 1
  \\
  \label{eq:brute-force-clause-3}
  \cidno{2}
  &&
     x_1 + \olnot{x}_2 +  x_3    &\geq 1
  \\
  \label{eq:brute-force-clause-4}
  \cidno{16}
  &&
     \olnot{y}_2 +  x_1 + \olnot{x}_2 + \olnot{x}_3  &\geq 1
  \\
  \label{eq:brute-force-clause-5}
  \cidno{1}
  &&
     \olnot{x}_1 +  x_2 +  x_3 &\geq 1
  \\
  \label{eq:brute-force-clause-6}
  \cidno{17}
  &&
     \olnot{y}_2 + \olnot{x}_1 +  x_2 + \olnot{x}_3 &\geq 1
  \\
  \label{eq:brute-force-clause-7}
  \cidno{18}
  &&
     \olnot{y}_2 + \olnot{x}_1 + \olnot{x}_2 +  x_3 &\geq 1
  \\
  \label{eq:brute-force-clause-8}
  \cidno{4}
  &&
     \olnot{x}_1 + \olnot{x}_2 + \olnot{x}_3 &\geq 1
\end{align}
\end{subequations}
to get an overview of the clauses involved in the derivation
fixing~$y_2$ to false. 
In the proof log, we can write a single long \verb+p+-line
\begin{Verbatim}[commandchars=\\\{\}]
p 15 3 + 2 d 2 16 + 2 d + 2 d 1 17 + 2 d 18 4 + 2 d + 2 d + 2 d
\end{Verbatim}
to implement the procedure described at the end of
Step~2 in 
\refsec{sec:proof-logging-translation}.
repeating derivations of the
clause~\refeq{eq:cl-not-assmt-xk-eliminated}
from~\refeq{eq:cl-not-assmt-incl-y-xk-false}
and~\refeq{eq:cl-not-assmt-incl-y-xk-true}
for partial assignments over subsets of variables of decreasing size.
First, the variable~$x_3$ is eliminated by performing addition
followed by division by~$2$  
for 
the clause pair~\refeq{eq:brute-force-clause-1}
and~\refeq{eq:brute-force-clause-2},
the pair~\refeq{eq:brute-force-clause-3}
and~\refeq{eq:brute-force-clause-4},
the pair~\refeq{eq:brute-force-clause-5}
and~\refeq{eq:brute-force-clause-6},
and the pair~\refeq{eq:brute-force-clause-7}
and~\refeq{eq:brute-force-clause-8},
respectively.
This yields four new clauses, for which addition followed by division
is performed in the same order to eliminate~$x_2$.
In the final step, the two clauses
$\olnot{y}_2  + x_1 \geq 1 $
and~$\olnot{y}_2 + \olnot{x}_1 \geq 1 $
are added and the result divided by~$2$ to yield the clause
\begin{align}
  \label{eq:y2-false}
  \cidno{19}
  &&
     \olnot{y}_2  &\geq 1 
\end{align}
as desired. (A further slight optimization could be to only add the
clauses together, without any intermediate division steps, and then
finally divide by a large enough number---the number of clauses
involved in the brute-force derivation will always be enough---but we
opted here for keeping all intermediate constraints clausal for
simplicity.) 

The constraint~\refeq{eq:y2-false}
can then be added to~\eqref{eq:example_pre_parity_2}
to remove $\vary_2$.
To eliminate~$\vary_2$ from~\eqref{eq:example_pre_parity_1} we can simply
use the literal axiom $\vary_2 \geq 0$,
which as mentioned above is referred to as ``\verb|y2|'' in the
\verb+p+-rule. 
Repeating this in formal notation, the proof lines
\begin{Verbatim}[commandchars=\\\{\}]
p 13 y2 +
p 14 19 +
\end{Verbatim}
derive the inequalities
\begin{subequations}
\begin{align}
    \cidno{20} && \varx_1 + \varx_2 + \varx_3 + 2 \olnot{\vary}_1 &\geq 2
    \label{eq:example_parity_1}\\
    \cidno{21} && \olnot{\varx{}}_1 + \olnot{\varx{}}_2 + \olnot{\varx{}}_3 + 2 \vary_1 &\geq 3
    \label{eq:example_parity_2}
\end{align}
\end{subequations}
encoding an equality
$\varx_1 + \varx_2 + \varx_3 = 2 \vary_1$
of the form~\eqref{eq:xor-pb}.

For the second parity constraint
$\varx_2 \oplus \varx_3 \oplus \varx_4 = 1$, 
we perform analogous derivations steps
\begin{Verbatim}[commandchars=\\\{\}]
red +2 \verbtilde{y3} +1  x2 +1  x3 +1  x4 >= 2 ; y3 -> 0
red +2  y3 +1 \verbtilde{x2} +1 \verbtilde{x3} +1 \verbtilde{x4} >= 2 ; y3 -> 1
red +3 \verbtilde{y4} +1  x2 +1  x3 +1  x4 +2 \verbtilde{y3} >= 3 ; y4 -> 0
red +3  y4 +1 \verbtilde{x2} +1 \verbtilde{x3} +1 \verbtilde{x4} +2  y3 >= 3 ; y4 -> 1
p 24 22 2 * + 3 d
p 25 23 2 * + 3 d
\end{Verbatim}
to obtain
\begin{subequations}
\begin{align}
  \cidno{26} 
  &&
     \varx_2 + \varx_3 + \varx_4 + 2 \olnot{\vary}_3 + \olnot{\vary}_4 
  &\geq 3
    \label{eq:example_2nd_pre_parity_1}\\
  \cidno{27} 
  &&
     \olnot{\varx}_2 + \olnot{\varx}_3 + \olnot{\varx}_4 + 
     2 \vary_3 + \vary_4 
  &\geq 3
    \label{eq:example_2nd_pre_parity_2}
\end{align}
\end{subequations}
after which we fix~$y_4$ to true by writing
\begin{Verbatim}[commandchars=\\\{\}]
rup +1 y4 +1  x2 +1  x3 +1 \verbtilde{x4} >= 1 ;
rup +1 y4 +1  x2 +1 \verbtilde{x3} +1  x4 >= 1 ;
rup +1 y4 +1 \verbtilde{x2} +1  x3 +1  x4 >= 1 ;
rup +1 y4 +1 \verbtilde{x2} +1 \verbtilde{x3} +1 \verbtilde{x4} >= 1 ;
p 5 28 + 2 d 29 6 + 2 d + 2 d 30 7 + 2 d 8 31 + 2 d + 2 d + 2 d
\end{Verbatim}
yielding the constraint
\begin{align}
  \cidno{32}
  &&
     y_4  &\geq 1 
\end{align}
on the last line. 
We finally derive the pseudo-Boolean constraints
\begin{subequations}
\begin{align}
    \cidno{33}
  && 
     \varx_2 + \varx_3 + \varx_4 + 2 \olnot{\vary}_3
  &\geq 3
    \label{eq:example_parity_1_2}\\
  \cidno{34} 
  &&
     \olnot{\varx{}}_2 + \olnot{\varx{}}_3 +
     \olnot{\varx{}}_4 + 2 \vary_3 
  &\geq 2
  \label{eq:example_parity_2_2}
\end{align}
\end{subequations}
encoding the PB equality
$\varx_2 + \varx_3 + \varx_4 = 2 \vary_3 + 1$
by the derivation steps
\begin{Verbatim}[commandchars=\\\{\}]
p 26  32 +
p 27 \verbtilde{y4} + 
\end{Verbatim}
and it is straightforward to verify that the constraints
\refeq{eq:example_parity_1_2}--\refeq{eq:example_parity_2_2}
indeed enforce that the parity of the variables
$x_2, x_3, x_4$
is odd.
This concludes the proof logging done after detecting parities.
We remark that in the  implementation of SAT solving with Gaussian
elimination that we made for the purposes of the experiments in this paper, 
the detection of parities and the proof generation for
pseudo-Boolean constraints encoding such parities is done only once at the
start of the solver execution. In principle, however, similar
detection and derivation steps could also be performed later during
the solver search.

Suppose now that that the solver decides on the assignment 
$\varx_1 = 0$. 
Note that adding the two parity constraints
$\varx_1 \oplus \varx_2 \oplus \varx_3 = 0$
and
$\varx_2 \oplus \varx_3 \oplus \varx_4 = 1$
encoded by our input formula yields
\mbox{$\varx_1 \oplus \varx_4 = 1$},
and hence $\varx_4$ should propagate to $1$.
This will be detected when the XOR propagator runs Gaussian
elimination.  

In order to justify this propagation, in the proof file
the solver first needs to derive the new parity constraint by adding
pairwise the pseudo-Boolean inequalities encoding the original parity
constraints, which is done by inserting the lines
\begin{Verbatim}[commandchars=\\\{\}]
p 20 33 +
p 21 34 +
\end{Verbatim}
producing the new constraints
\begin{subequations}
\begin{align}
\label{eq:new-parity-1}
\cidno{35} && 
\varx_1 + 
\varx_4 + 
2\varx_2 + 2\varx_3 +
2\olnot{\vary}_1  + 2 \olnot{\vary}_3 &\geq 5 
\\
\label{eq:new-parity-2}
\cidno{36} && 
\olnot\varx_1 + 
\olnot\varx_4 + 
2\olnot\varx_2 + 2\olnot\varx_3  +
2{\vary}_1  + 2 {\vary}_3 
&
\geq 5
\end{align}
\end{subequations}
that imply $\varx_1 \oplus \varx_4 = 1$ by
the observation made in
in \refsec{sec:proof-logging-pb-encoding}.

Once the constraints
\refeq{eq:new-parity-1}--\refeq{eq:new-parity-2}
have been added to the constraints database of the verifier,
the solver also needs to provide a proof that
the reason clause 
${\varx_1 + \varx_4 \geq 1}$ 
provided by the XOR propagator is valid.
The assignment falsifying this reason clause is 
$\s{\rho} = \set{\varx_1\mapsto 0, \varx_4 \mapsto 0}$.
Following the approach in
\refsec{sec:proof-logging-clause-generation},
we  derive $\varx_1 + \varx_4 \geq 0$ and add to
constraint~$36$ in~\refeq{eq:new-parity-2}
to get
$
2\olnot\varx_2 + 2\olnot\varx_3  + 2{\vary}_1  + 2 {\vary}_3 
\geq 3
$,
after which division by~$2$ followed by multiplication by~$2$ yields
$
2\olnot\varx_2 + 2\olnot\varx_3  + 2{\vary}_1  + 2 {\vary}_3 
\geq 4
$.
If we add
constraint~$35$ in~\refeq{eq:new-parity-1} to this, then
the terms
$2\olnot\varx_2 + 2\olnot\varx_3  + 2{\vary}_1  + 2 {\vary}_3 $
and
$ 2\varx_2 + 2\varx_3 + 2\olnot{\vary}_1  + 2 \olnot{\vary}_3 $
cancel, leaving a constant~$8$, and so if we write the line
\begin{Verbatim}[commandchars=\\\{\}]
p 36 x1 x4 + + 2 d 2 * 35 +
\end{Verbatim}
in the proof log, then this yields the clause
\begin{align}
\label{eq:example-final-reason-clause}
\cidno{37} && 
\varx_1 + \varx_4 &\geq 1
\end{align}
proving that the propagation is valid.
Observe that in contrast to the other constraints derived in the proof
logging steps above, the reason clause
$\varx_1 \lor \varx_4 $
in~\refeq{eq:example-final-reason-clause}
is also stored in the SAT solver clause database and can can be used in the
ensuing CDCL search in the same way as any other clause in this
database.

Whenever the XOR propagator detects a propagation or conflict, 
the solver will need to write derivation steps analogous to the ones
leading to constraints
\refeq{eq:new-parity-1}--\refeq{eq:new-parity-2}
and
\refeq{eq:example-final-reason-clause}
to the proof file. After this, the
clause~\refeq{eq:example-final-reason-clause}
can be used either for propagation or as the starting point for CDCL
conflict analysis.

\section{Implementation and Evaluation}
\label{sec:evaluation}

We have extended the pseudo-Boolean proof format (\pbp)
\ifthenelse{\boolean{detectedJAIR}}
{of the \veripb tool%
  \footnote{\veripb is available
    at~\url{https://gitlab.com/MIAOresearch/VeriPB}.} 
  with a}
{of the \veripb tool~\cite{VeriPB} with a}
\redbasedstrengthening rule,
which the proof checker validates as described in
Algorithm~\ref{alg:PBRAT}, 
and have
implemented our proof logging approach for XOR reasoning in a
library 
together with an XOR engine
using Gaussian elimination 
$\bmod\  2$ 
to detect XOR propagations.%
\urldef{\xorengineurl}\url{https://gitlab.com/MIAOresearch/xorengine}%
\footnote{The code for the XOR engine is available at \xorengineurl{}.}
We integrated this library into
\ifthenelse{\boolean{detectedJAIR}}
{\minisat\footnote{\minisat is available at \url{http://minisat.se/}.}} 
{\minisat~\cite{MiniSat}}
to call the XOR propagation method every time
clausal propagation terminated.
If the library detects a propagation or
conflict, a callback is used to notify 
\minisat, but the reason clause is only generated
when needed in conflict analysis. 
This \emph{lazy reason generation} technique~\cite{SGM20Tinted}
is crucial to minimize the proof logging overhead,
since it avoids generating proofs 
for reasons that are not used.

In order to be able to compare to approaches using \drat, 
we have also implemented in our library  \drat proof logging for XOR
constraints as described \jnciteinorby{PR16DRATforXOR}.
We remark that we did not study the more recent \drat-based approach
\jnciteinorby{CH20SortingParity}, which combines long parity constraints by
sorting the involved literals, because it does not seem to be
applicable to the kind of formulas that are relevant for our
comparison with \drat. 
The parity constraints in the formulas we consider will only contain
few variables, or else the clausal encoding that we are looking for to
detect these parities will blow up the formulas exponentially.
Also, when we operate on intermediate parity constraints generated
during Gaussian elimination, such
parities are guaranteed to be sorted already.

In the results reported below, 
all running times were measured on an 
Intel  Core
i5-1145G7 @2.60GHz $\times$ 4 with a
memory limit of~8GiB, disk write speed of roughly 200 MiB/s, and read speed of
2 GiB/s. The used tools, benchmarks, data and evaluation scripts are
available at \url{https://doi.org/10.5281/zenodo.7083485}.

Importantly, our goal was not to 
study whether XOR reasoning is  useful or not---this has already been
investigated---but  
to provide efficient proof logging for such reasoning. Therefore, we focused on 
benchmarks from the 
\jncitenameurl{SAT competition}{SATcompetition}{}
from 2016 to 2020 that could be solved by
\minisat with our XOR propagator but not by 
\jncitenameurlpunctuation{\kissat}{Kissat}{http://fmv.jku.at/kissat/}{,}
the winner of the 2020    SAT  competition. There were 
$39$~such instances, and they could be solved in $0.01$ seconds on
average by \minisat{} with the XOR propagator.
With our new proof logging
the average running time increased to $0.02$~seconds and
unsatisfiability could be verified 
in $1.29$~seconds on average. For
\drat proof logging, on the other hand, the average solving time jumped
to $2.72$~seconds and verification 
took $1092$~seconds on average.

In order to get systematic measurements for the performance of our
new proof logging technique, we ran experiments on the so-called
\emph{Tseitin formulas}%
\footnote{Somewhat confusingly, and as can be seen from the instance 
  names in Table~\ref{tbl:proofs}, these formulas are sometimes also
  referred to as \emph{Urquhart formulas} in the applied SAT
  community, perhaps because 
  \jnciteNameRef{Urquhart87HardExamples}{Urquhart}
  was the first to establish
  strong hardness results for these formulas.
} 
introduced \jnciteinorby{Tseitin68ComplexityTranslated},
including some formula instances that
have been studied before 
in the applied SAT community
in the context of proof logging.
Tseitin formulas consist of large inconsistent sets of parity constraints,
and can thus be viewed as a worst case for XOR reasoning.
To the best of our knowledge, the shortest \drat proofs
for these formulas obtained so far%
\footnote{The proofs and instances can be found at
  \url{https://github.com/marijnheule/drat2er-proofs}.}
are based on hand-crafted so-called
\textit{propagation redundancy} (PR) proofs, which have been translated
to \drat using the tool \prtodrat~\cite{KRH18DRAT2ExRes}.
Table~\ref{tbl:proofs} shows the disk space required for the proofs of
Tseitin formulas 
\jnciteinorby{KRH18DRAT2ExRes}.
The
pseudo-Boolean
proofs obtained by \minisat{} with the XOR propagator are dramatically 
smaller than the 
\drat 
proofs 
produced by the same tool,
and the size of our 
\drat 
proofs are
similar to that of the best previously known 
\drat 
proofs.

\begin{table}[t]
    \centering
    \newcolumntype{R}{>{\raggedleft\arraybackslash}X}
    \begin{tabular}{lrrr}
    \hline
    Instance & \multicolumn{2}{c}{\minisat{} + XOR} & \prtodrat \\ 
             & (\pbp) & (\drat) & \\   
    \hline
    Urquhart-s5-b1 & 80.8 & 3033.1 & 3878.4 \\ 
    Urquhart-s5-b2 & 84.0 & 2844.4 & 3575.2 \\ 
    Urquhart-s5-b3 & 123.5 & 7584.0 & 7521.0 \\ 
    Urquhart-s5-b4 & 99.8 & 5058.6 & 5271.5 \\ 
    \hline
    \end{tabular}
    \caption{Proof sizes (KiB) for
   some previously studied
      Tseitin formulas.}
  \label{tbl:proofs}
\end{table}

\begin{figure}[t]
\centering
        \includegraphics[scale=0.8]{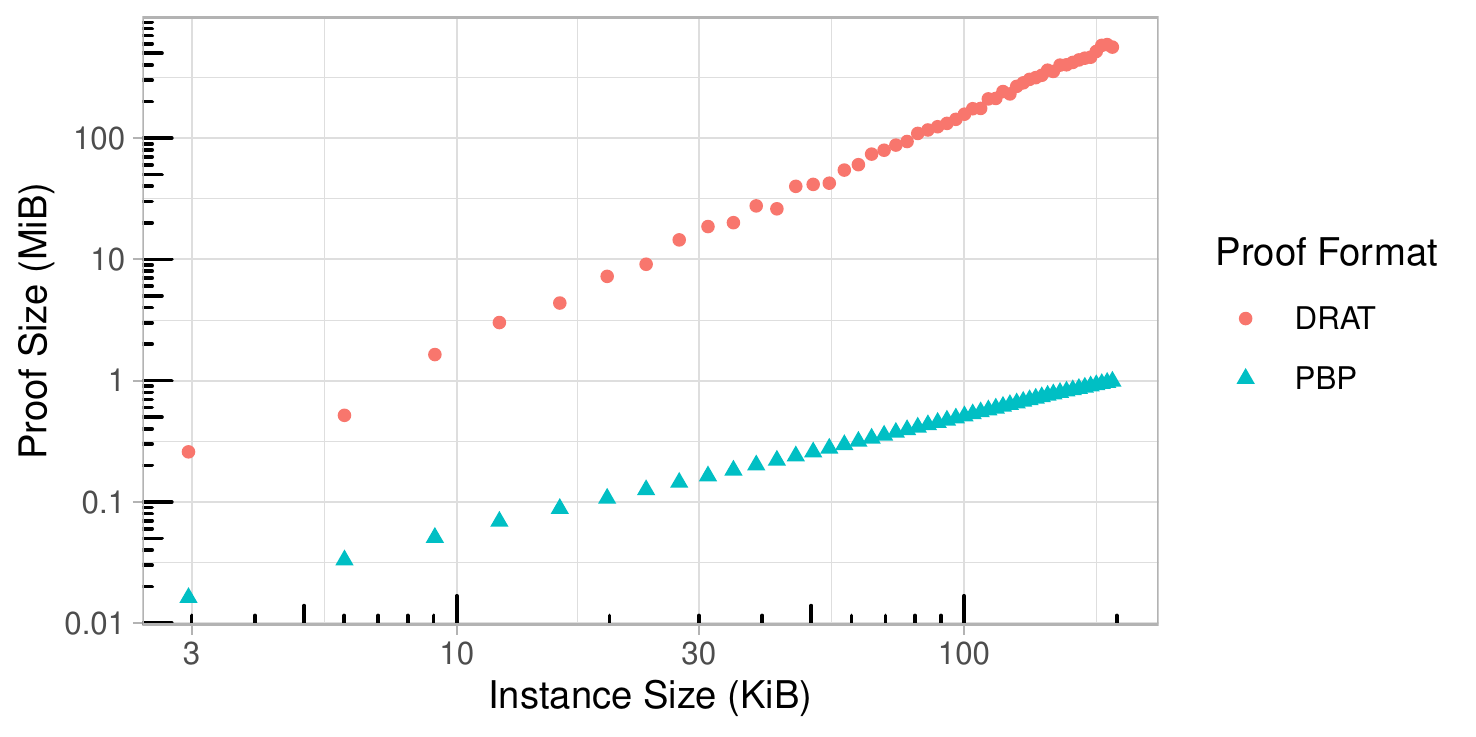}
        \caption{Proof sizes for larger Tseitin formulas
          using \drat and 
          PB
          proof logging.}
        \label{fig:tseitin_proof_size}
\end{figure}
The formulas in Table~\ref{tbl:proofs} contain only $50$~XOR
constraints over about $100$~variables, which is very small by modern
standards, and they are
solved and verified in less than one second.
To get a sense
of the asymptotic behaviour of the
proof logging, we 
also considered
$50$~new,
larger
Tseitin formulas with up to $500$~XORs
over
up to $1250$~variables. 
These formulas were 
generated with the tool
\cnfgen~\jnshortcite{LENV17CNFgen} using random 
regular graphs of 
degree~$5$,
which produces formulas with clausal encodings of
\mbox{$5$-XOR} constraints.
It was shown \jnciteinorby{Urquhart87HardExamples} that Tseitin formulas
are hard for \emph{resolution}, 
the  reasoning method
underlying conflict-driven clause learning,
if the graph 
from which the formula is generated   %
is an expander,
and it is well known
that random graphs are expanders 
with extremely high probability
(see, e.g.,~\cite{HLW06ExpanderGraphs}). Thus,
we can be confident that 
the generated formulas are hard for CDCL solvers and 
require additional reasoning methods, such as Gaussian elimination, 
to be solved efficiently.

In \reffig{fig:tseitin_proof_size} we compare
the proof size for \drat proof logging and 
our pseudo-Boolean \veripb proof logging.
Notice that
both proof logging approaches result in 
straight lines
in the log-log
plot, which is a strong indication
that they are both scaling polynomially. 
Studying the slopes of the lines yields the estimates that \drat
produces quadratic-size proofs while the proof size of the
pseudo-Boolean proof is linear in the size of the formula.
In \reffig{fig:tseitin_time} we compare the running time (system
time 
plus  %
user time) of solving and producing the proof, as well as time 
spent on proof verification
(where it can be noted that running times below one second should be
interpreted with some care since the running time might be dominated by start-up
overhead).   
It is clear that the larger proof size required for \drat proofs does
not only increase verification time, but also causes a clearly
increased time overhead during solving.

\begin{figure}[t]
    \centering
        \includegraphics[scale=0.8]{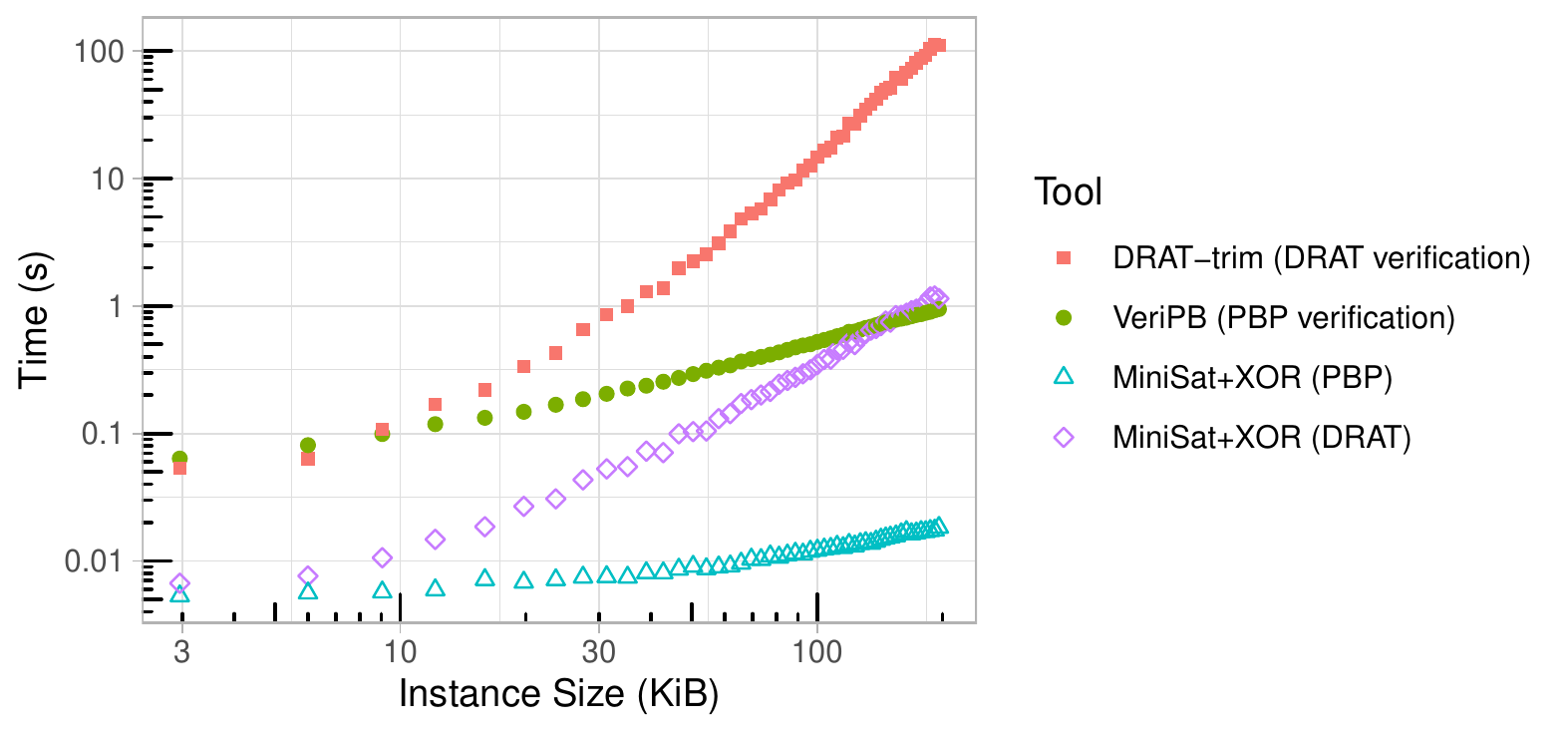}
        \caption{Solving and verification time for Tseitin formulas.}
        \label{fig:tseitin_time}
\end{figure}

\ifthenelse{\boolean{detectedJAIR}}
{
\begin{figure}[t]
\centering
        \includegraphics[scale=0.8]{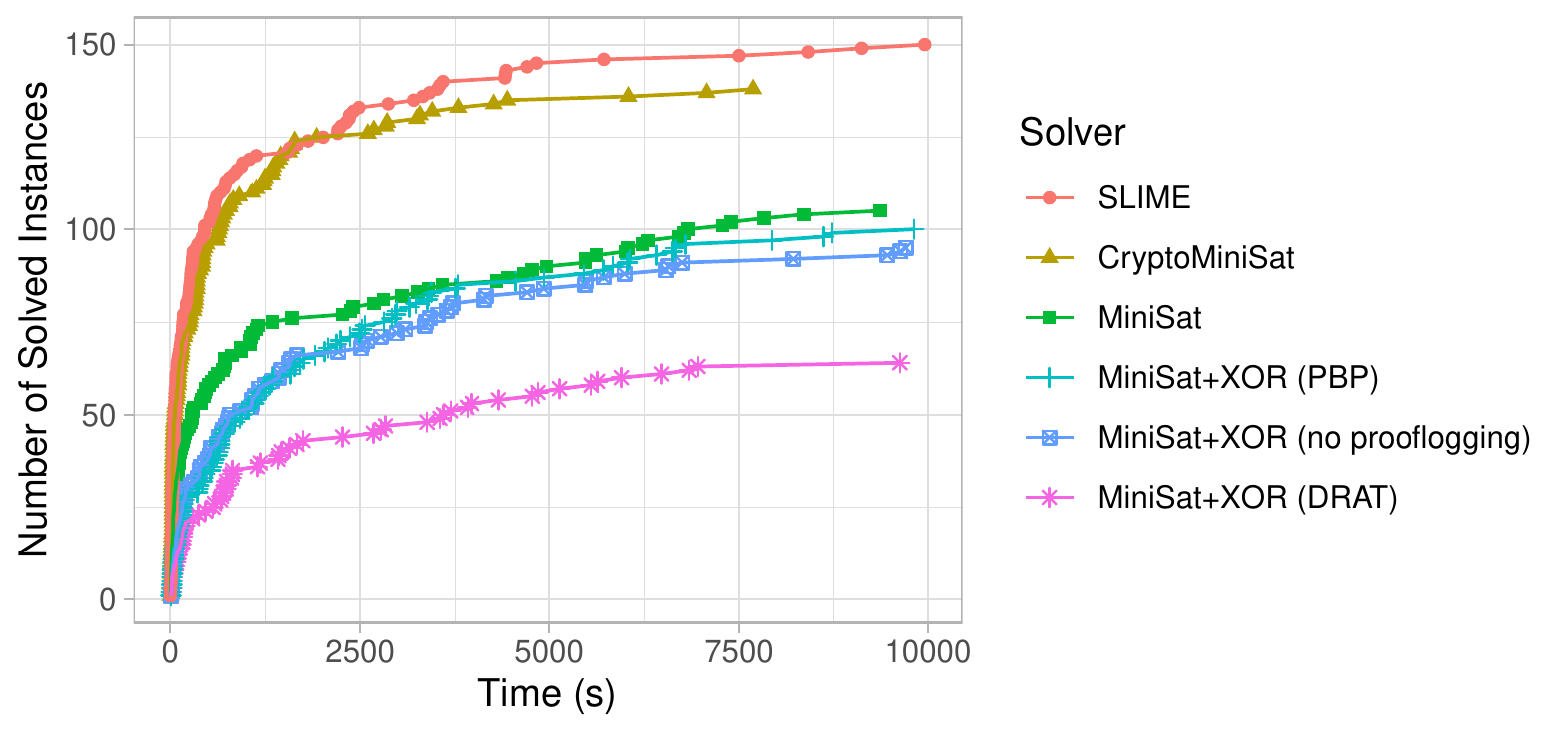}
        \caption{%
          Cumulative plot 
          for
          the crypto track of the SAT Competition 2021.
        } 
        \label{fig:crypto_comulative}
\end{figure}
}
{
\begin{figure}[t]
\centering
        \includegraphics[scale=0.8]{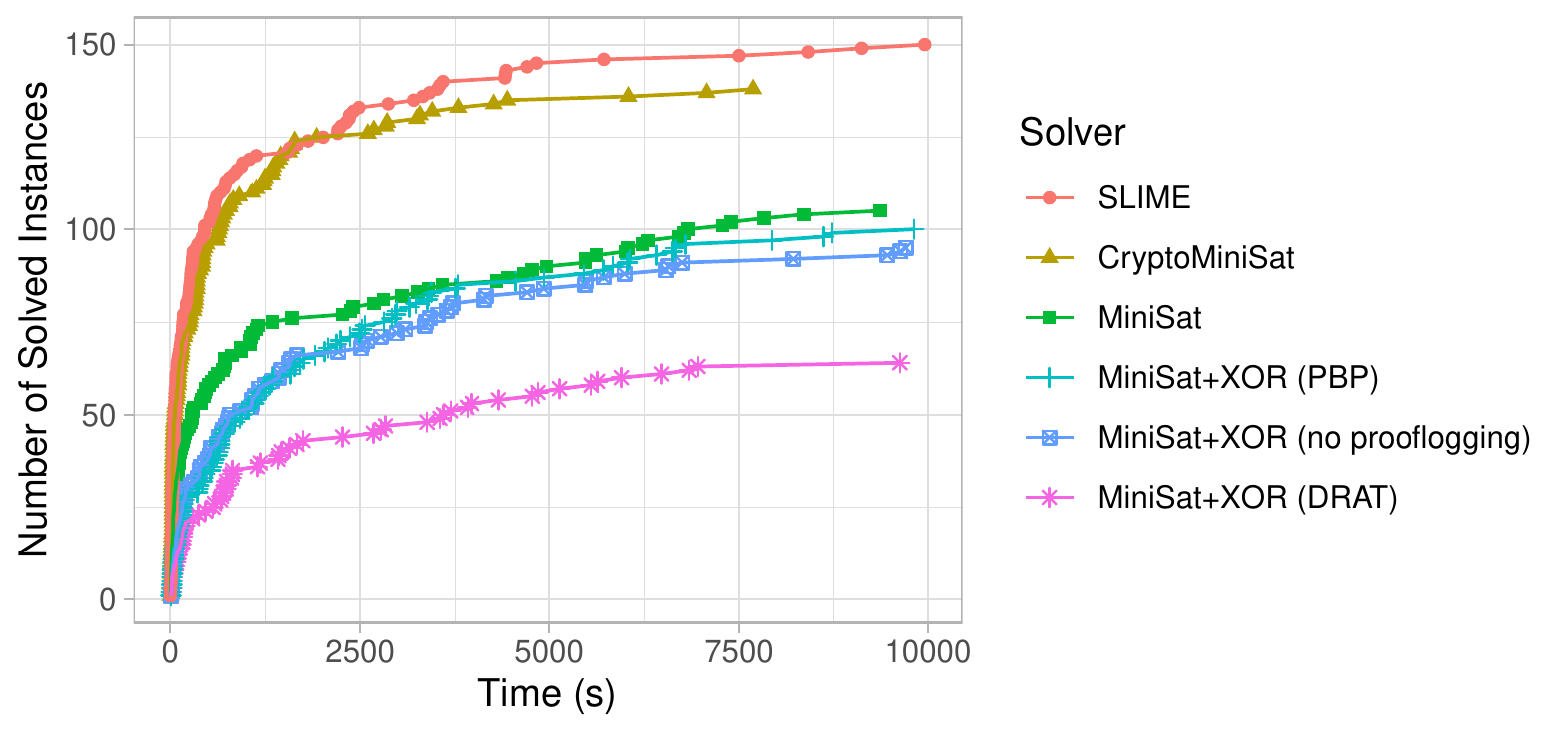}
        \caption{%
          Cumulative plot 
          for number of solved instances in
          the crypto track of the SAT Competition 2021.
        } 
        \label{fig:crypto_comulative}
\end{figure}
}

To get a wider range of
practically relevant formulas, 
we 
additionally
evaluated our tools on
cryptographic benchmarks, which often contain parity constraints, from the
crypto track of the 2021 SAT competition. 
\reffig{fig:crypto_comulative} compares the performance of different
solvers on this 
benchmark set, including 
\jncitenameurlpunctuation{\slime}{SLIME}{https://maxtuno.github.io/slime-sat-solver/}{,}
the winner of the crypto track,
and 
\jncitenameurlpunctuation{\cryptominisat}{CryptoMiniSat}{https://github.com/msoos/cryptominisat/}{,}
arguably the most well-established modern solver with 
integrated parity reasoning.
Notably, \slime and \cryptominisat 
significantly
outperform \minisat, showing the
advancements made over the last decades. Somewhat surprisingly, our
integration of parity reasoning does not seem to benefit \minisat on this set
of benchmarks. However, if one insists on
that the solver with parity reasoning should also support proof logging, 
then it is clear that more instances can be solved 
if we use pseudo-Boolean proof logging instead of \drat.
One reason for this
is that the proof sizes are much
larger when using \drat proof logging, as
 shown in 
\reffig{fig:crypto_proof_size}. The 
generated
\drat
proofs can quickly exceed the disk limit of
roughly 100GB, causing the 
SAT 
solver to terminate with an error. 

\begin{figure}[t]
\centering
        \includegraphics[scale=0.8]{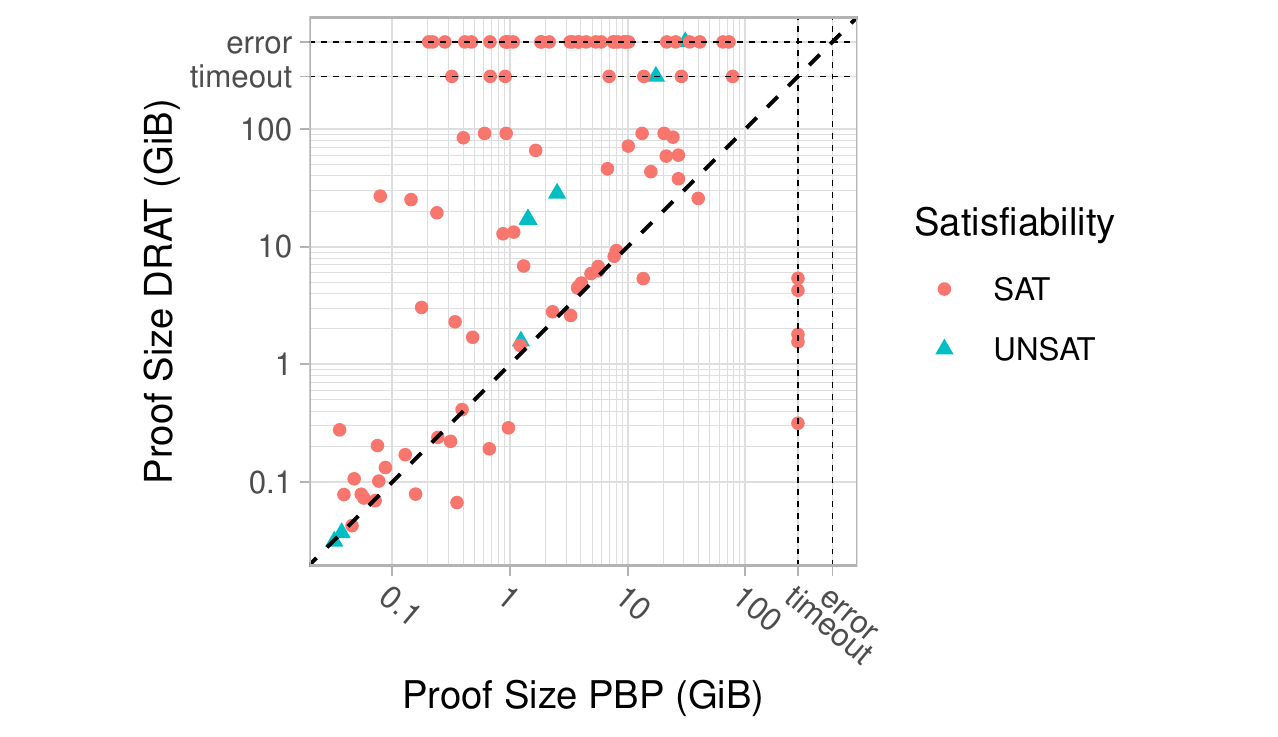}
        \caption{%
          \drat and PB and proof sizes for 
         the crypto track of the SAT Competition 2021.}
        \label{fig:crypto_proof_size}
\end{figure}

While the
tendency of 
the    %
plot 
in     %
\reffig{fig:crypto_proof_size} is clear, it should be
noted that the difference shown 
is due not only to different proof logging methods, 
but also to the particular way in which we implemented proof logging for
\minisat, in which the introduction of new variables for proof logging affects
the  \minisat search. This can be observed in different statistics such as
the number of decisions or  conflicts.%
\footnote{In principle, the CDCL proof search should be completely
  oblivious to whether proof logging is being carried out or not, since no
  proof logging steps have any bearing on how the search algorithm is
  executed. However, in our implementation we use the variable
  handling interface in \minisat to manage the auxiliary variables introduced
  during proof logging. In more technical detail, the proof logging
  routines introduces fresh variables by adding them to the solver and
  marking them as non-decision variables.
  The mere existence of these additional variables seems to cause a
  slight change in the search. The difference 
  can only be observed when
  variables were added before preprocessing.
  With hindsight, it would most likely be better to let the proof
  logging code manage additional variables only used for the
  pseudo-Boolean derivations separately from the solver,
  However, our ambition was not to deliver a production-grade SAT
  solver with Gaussian elimination, but to provide a competitive
  implementation that can serve as a basis for meaningful experiments.
}
For example, consider the instance
in the bottom right of \reffig{fig:crypto_proof_size} that requires a
proof of a few hundred MiB in 
\drat   %
but times out for pseudo-Boolean proof
logging. This instance is solved with 
$541,928$ conflicts 
with no proof
logging
or   %
\drat   %
proof logging, but requires more than $19$~million conflicts
for pseudo-Boolean proof logging. 
It should be emphasized, however, that this difference is completely
irrelevant in that it is not in any way related to pseudo-Boolean
proof logging per se, but only to a peculiar choice in the
implementation we used for our experiments, as explained in the footnote above.

In \reffig{fig:crypto_verification} we
can see the time required for solving a benchmark versus verifying the
result. In practice, it would be sufficient to only verify the final solution
for satisfiable instances. However, as there are few solved unsatisfiable
instances, we verified that every constraint derived by the solver is correct
even for satisfiable 
ones.
For most 
formulas
the verification overhead
is roughly a factor~$10$, but there are also 
cases
where verification is
much slower, which demonstrates that further 
optimizations in the \veripb proof verification code would be desirable.
We are aware of several possibilities for this, such as improving
RUP checks and borrowing ideas like backward trimming of proofs from
\jncitenameurlpunctuation{\drattrim}{DRATtrim}{https://github.com/marijnheule/drat-trim}{.}
However, the main focus of this work was not on 
such   %
engineering questions, but rather on developing new mathematical
methods for efficient proof logging, and we would argue that the
potential for vast improvements in proof logging efficiency should be
clear from the experimental results reported in this section.

\begin{figure}[t]
\centering
        \includegraphics[scale=0.8]{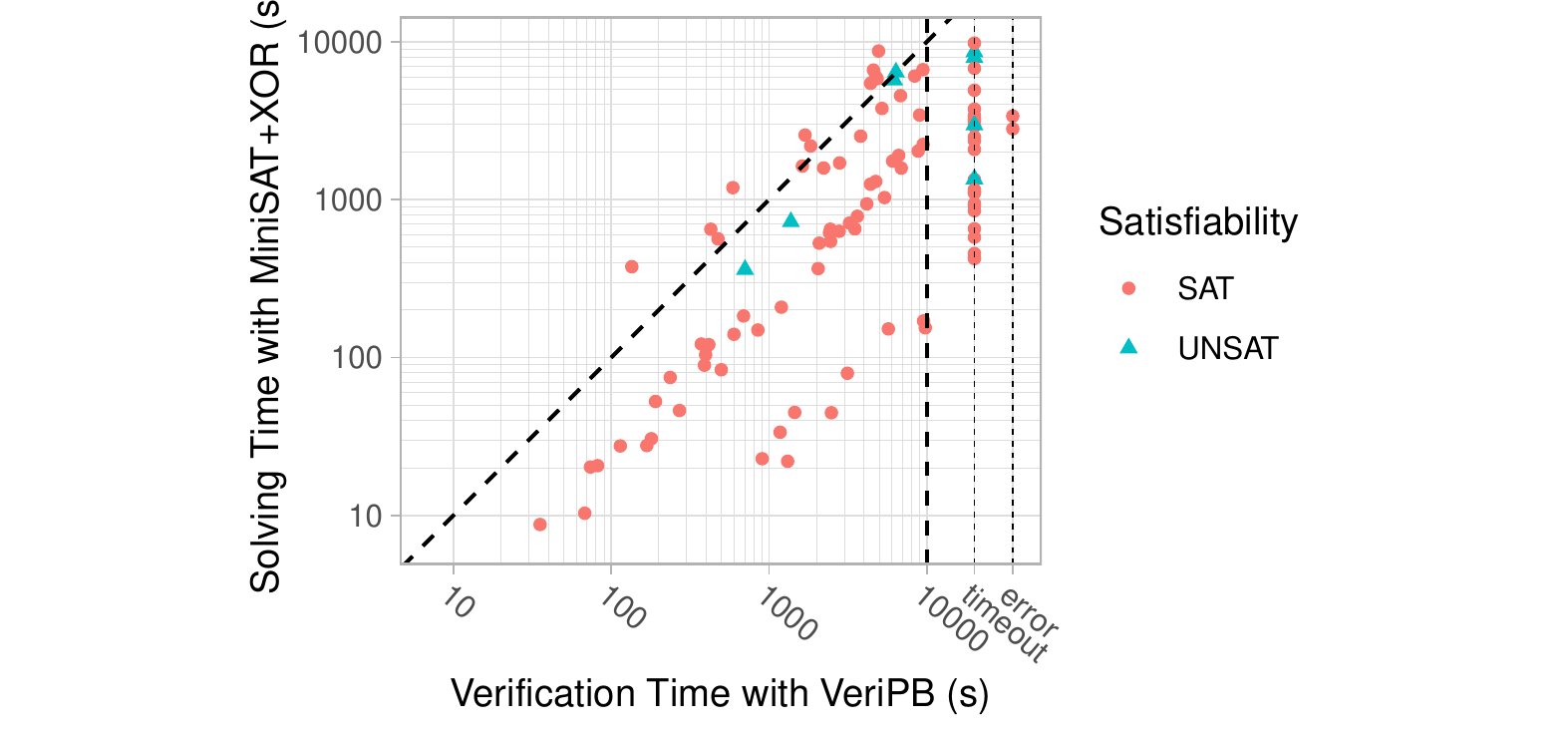}
        \caption{Time required for solving and verifying instances.}
        \label{fig:crypto_verification}
\end{figure}

\section{Conclusion}
\label{sec:conclusion}

In this work, we present 
an efficient proof logging method for
conflict-driven clause learning (CDCL) solvers
equipped with parity reasoning, 
which has been a long-standing challenge in SAT solving.  
Our 
approach
circumvents
the prohibitive overhead of 
previous
\drat-based proof logging methods for 
{parity reasoning such as the one developed \jnciteinorby{PR16DRATforXOR}}
by instead using the cutting planes method operating on pseudo-Boolean 
inequalities in the 
\jncitenameurl{\veripb tool}{VeriPB}{} and adding a rule for
introducing extension variables.
An experimental evaluation shows that this makes the proof logging
overhead, the size of the proof, and the time required for
verification all go down by an order of magnitude or more compared to
\drat.
While there is certainly ample room for further improvements, our
first proof-of-concept implementation already shows the power of this
approach. 

In terms of 
weaknesses,
one significant disadvantage
of our method
is that
the proof verification time is still considerably larger than the time
required for solving with proof logging, especially if many XOR
constraints are involved. There are at least two explanations for this.
One reason is that the algorithm for XOR reasoning can make use of bit-level
parallelism. The verifier cannot do so easily,
because it has to be able to deal with arbitrary linear constraints
and not just XORs. Another reason is that we introduce fresh variables
to encode the XOR constraints. 
On the solver side, these auxiliary variables can essentially be ignored,
except that they are printed in fairly standardized proof logging
templates, but they play
a crucial role in the calculations on the proof checker side
when  the  proof is verified.
It should be said, though, that although verification overhead is
larger than proof logging overhead, this is only by a constant factor.
In other words, if we are willing to pay a constant-factor increase in
running time, then this will allow us to not only use parity reasoning
but also obtain a formal proof establishing that this parity reasoning has
been performed correctly. It seems fair to argue that the benefits from fully
verified solutions could outweigh the disadvantage of this limited increase
in total execution time.

By construction, the pseudo-Boolean proof logging method in \veripb
can also be used to solve another
task that has remained very challenging for \drat, namely efficient
proof logging for cardinality detection and reasoning. We have not
investigated this in 
the current 
paper, since this is mostly an engineering
question
rather than a research problem
in view of 
the methods that have
already been developed 
\jnciteinorby{BLLM14Detecting,EN20CardinalImprovement}. 
Symmetry handling, a third notorious problem for proof logging,
appears to be much more difficult, but when it comes to adding
symmetry breaking constraints  our method can do at least as well
as~\jnciteA{HHW15SymmetryBreaking}, since it is a strict generalization of 
\drat.
In a later work~\cite{BGMN22Dominance} appearing after the conference
version of this paper, the \veripb proof logging system has been
extended further with a so-called 
\emph{dominance-based strengthening}, providing for the first time
efficient proof logging support for fully general symmetry breaking.
It is an interesting open question, however, whether this new
dominance rule is necessary, or whether the \redbasedstrengthening
rule introduced in the current work is sufficient to provide efficient
derivations of symmetry-breaking constraints.

The fact that no efficient proof logging support has previously been
available for enhanced SAT solving techniques such as parity
reasoning, cardinality detection, and symmetry handling means that SAT
solvers making crucial use of such techniques have not been able to
take part in the main track of the 
\jncitenameurlpunctuation{SAT competition}{SATcompetition}{}{,} 
where proof logging is mandatory.
Somewhat paradoxically, this seems to have the effect that the proof
logging requirements, which have played such an important role for the
development of the field, now risk becoming a barrier to further
solver developments. Since the \veripb tool can now support all of
parity reasoning, cardinality detection, and (as
of~\cite{BGMN22Dominance}) also symmetry breaking,
and does so with very limited overhead compared to \drat, it seems
natural to propose that this should be an allowed proof logging format
in future SAT competitions.

However, we believe that the potential benefit of pseudo-Boolean proof
logging with extension variables goes well beyond the context of the SAT
competitions. 
\veripb has been shown to be capable of efficient justification of
important constraint programming techniques~%
\cite{EGMN20Justifying,GMN22AuditableCP}, 
and can also provide proof logging for a wide range of graph problem
solvers~\cite{GMN20SubgraphIso,GMMNPT20CertifyingSolvers}.
Furthermore, 
\ifthenelse{\boolean{detectedJAIR}}
{\jnciteA{GMNO22CertifiedCNFencodingPB,VWB22QMaxSATpb}}
{the papers~\cite{GMNO22CertifiedCNFencodingPB,VWB22QMaxSATpb}}
have used \veripb to develop proof logging methods that seem to have
the potential to support a range of SAT-based optimization approaches
using maximum satisfiability (MaxSAT) solvers.
The pseudo-Boolean rules for reasoning with \mbox{$0$-$1$} linear
constraints provide a simple yet very expressive formalism, and it
does not seem out of the question to hope that they could be extended
to deal with proof logging for mixed integer programming (MIP).  Thus,
we believe that the ultimate goal of this line of research should be
to design a unified proof logging approach for as wide as possible a
range of combinatorial optimization paradigms.
In addition to furnishing efficient machine-verifiable proofs of
correctness, proof logging could also serve as a valuable tool for
debugging and empirical performance analysis during solver
development.  Furthermore, the proofs produced could in principle
provide auditability by third parties using independently developed
software, and/or be a stepping stone towards explainability by
showing, e.g., why certain solutions are optimal.
We view our paper as only one of the first steps on this long but
exciting road.

\section*{Acknowledgments}

We are grateful to Bart Bogaerts and Ciaran McCreesh for many
stimulating conversations on proof logging in general and \veripb in particular.
We also want to thank  Kuldeep~Meel and  Mate~Soos for helpful
discussions on how to implement Gaussian elimination modulo $2$.
A special thanks goes to Andy~Oertel for helping us track down mistakes 
in our worked-out example in \refsec{sec:example}.
Finally, we have benefited greatly from the interactions with and
feedback from many colleagues  taking part in the semester program
\emph{Satisfiability: Theory, Practice, and Beyond}
in the spring of 2021 at the 
Simons Institute for the Theory of Computing at UC Berkeley.

The authors were supported by the Swedish Research Council grant
\mbox{2016-00782}, and Jakob~Nordström also received funding from the
Independent Research Fund \mbox{Denmark} grant \mbox{9040-00389B}.

\bibliography{refArticles,refBooks,refOther,paper}

\bibliographystyle{alpha}

\end{document}